\documentclass[12pt, draftclsnofoot, onecolumn]{IEEEtran}
\usepackage[super]{nth}
\usepackage{amsthm}
\usepackage{algorithm}
\usepackage{amssymb}
\usepackage{multirow}

\usepackage{tabularx}  
\usepackage{arydshln,leftidx,mathtools} 

\theoremstyle{definition}
\newtheorem{theorem}{Theorem}

\newtheorem{definition}{Definition}

\usepackage{cite}
\ifCLASSINFOpdf
 \usepackage[pdftex]{graphicx}
\else
\fi

\usepackage{amsmath}
\usepackage{array}
\usepackage[tight,footnotesize]{subfigure}
\usepackage{courier}
\usepackage[usenames,dvipsnames]{color} 
\usepackage{color,url,xspace}
\usepackage{amssymb,amsmath}
\usepackage{graphicx}
\usepackage{colordvi}
\usepackage{epsfig}
\usepackage{endnotes}
\usepackage{verbatim}
\usepackage{textcomp}
\usepackage{subfig}
\usepackage{setspace}
\usepackage{xcolor}
\usepackage{cite}
\usepackage{times}
\usepackage{booktabs, multirow}
\usepackage{amsthm}
\usepackage{nccmath}
\usepackage{algorithmicx,algpseudocode}
\usepackage{algpseudocode,algorithm,algorithmicx}
\usepackage{caption}
\usepackage{mathtools}
\usepackage{subfig}

\algrenewcommand\algorithmicrequire{\textbf{Precondition:}}
\algrenewcommand\algorithmicensure{\textbf{Postcondition:}}

\begin{document}
	
\title{Dynamic Power Allocation and User Scheduling for Power-Efficient and Low-Latency Communications}

\author{Minseok Choi,~\IEEEmembership{Member,~IEEE,}
	Joongheon Kim,~\IEEEmembership{Senior Member,~IEEE,}
	and~\\Jaekyun Moon,~\IEEEmembership{Fellow,~IEEE}
	\thanks{M. Choi and J. Moon are with the Department of Electrical Engineering, Korea Advanced Institute of Science and Technology (KAIST), Daejeon, Korea e-mails: ejaqmf@kaist.ac.kr, jmoon@kaist.edu.}
	\thanks{J. Kim is with the School of Computer Science and Engineering, Chung-Ang University, Seoul, Korea e-mail: joongheon@cau.ac.kr.}
}

\maketitle

\begin{abstract}
 
In this paper, we propose a joint dynamic power control and user pairing algorithm for power-efficient and low-latency hybrid multiple access systems.
In a hybrid multiple access system, user pairing determines whether the transmitter should serve a certain user by orthogonal multiple access (OMA) or non-orthogonal multiple access (NOMA). 
The proposed optimization framework minimizes the long-term time-average transmit power expenditure while reducing the queueing delay and satisfying time-average data rate requirements.
The proposed technique observes channel and queue state information and adjusts queue backlogs to avoid an excessive queueing delay by appropriate user pairing and power allocation.
Further, user scheduling for determining the activation of a given user link as well as flexible use of resources are captured in the proposed algorithm.
Data-intensive simulation results show that the proposed scheme guarantees an end-to-end delay smaller than 1 ms with high power-efficiency and high reliability, based on the short frame structure designed for ultra-reliable low-latency communications (URLLC).

\end{abstract}


\IEEEpeerreviewmaketitle

\section{Introduction}

The fifth-generation (5G) wireless networks are expected to offer high spectral efficiency, improved reliability, massive connectivity, and low end-to-end (E2E) latency \cite{5G:Andrews, 5G:Wunder}. 
Especially for realizing tactile Internet services, a very low E2E latency of 1\,ms should be guaranteed while providing reliable service quality \cite{Tactile:JSAC2016Simsek}.
In particular, International Telecommunication Union (ITU) has defined ultra-reliable low latency communication (URLLC) as one of usage scenarios in 5G networks \cite{ITU-5G}. 
With the proliferation of smart devices, particularly in the Internet of Things (IoT) network, URLLC should not only provide sufficiently high system throughput and low latency, but also support a massive scale of machine type communications \cite{URLLC:IEEENetwork2018Popovski,URLLC:arXiv2018Bennis}.
In addition, power-efficiency becomes critical especially for small and clumsy battery-powered IoT devices~\cite{PowerEff:JSAC2016Buzzi,EE:CommMag2014Hu,EE:IoTJ2015Pan}. 
Further, flexibility is also important to communicate with diverse machine type devices as well as human users while meeting a variety of quality of service (QoS) requirements \cite{CR:IoTJ2015Aijaz}.
Many researchers have studied a myriad of technical issues as mentioned above, and the works are in progress.

Actually, delay-constrained communication has long been a major challenge and interest. 
Given a delay constraint, the tradeoff between reliability and delay is studied in \cite{URLLC_reliab-delay:TC2018Nielsen}, and throughput analysis is also performed in \cite{URLLC_throughput:WoWMoM2016Chen}.
The packet delay can be reduced by designing a short frame structure \cite{URLLC_shortPacket:Proceeding2016Durisi,ShortPacket:VTC2015Kela} and/or adjusting the transmission policy \cite{URLLC_scheduling:TWC2017Zuo}.
The E2E delay consists of uplink (UL)/downlink (DL) transmission delays and queueing delay \cite{QueueDelay:VTC2016She}, and a short frame structure reduces UL/DL transmission durations.
Meanwhile, deterministic queueing delay analysis is generally considered difficult due to the fact that queue dynamics in medium access control (MAC) is influenced by the randomness of time-varying channels and stochastic geometry in physical (PHY) layer.
In \cite{QueueDelay:TWC2003Wu}, the \textit{effective capacity} link-layer model is presented to define the statistical delay requirement.
Based on the \textit{effective capacity} link model \cite{QueueDelay:TWC2003Wu}, cross-layer transmission design for achieving queueing delay requirements has been investigated under the assumption of a constant service rate and static channel over transmission time \cite{QueueDelay:GC2016She, QueueDelay:TWC2018She}.
Further, based on the Little's theorem \cite{Book:DataNetworks} which establishes that the time-average queueing delay is proportional to the average queue backlog, dynamic resource allocations and scheduling policies for reducing queueing delay have been actively pursued~\cite{URLLC_scheduling:TWC2017Zuo,delay-constrained:JSAC2015Khalek,delay-constrained:JSAC2018Choi,delay-constrained:arXiv2018Hsieh}.

Since there exists a fundamental power-delay tradeoff studied in \cite{TIT2002Berry,TIT2013Berry}, power efficiency also becomes critical for URLLC, especially where a massive number of devices are battery-powered \cite{PowerEff:GC2017Sun}.
Energy-efficient resource allocations and scheduling policies for delay-constrained communications have been studied in \cite{Infocom2003Fu,TWC2009Lee,WCNCW2013Ren,TIT2006Bettesh} and a delay-optimal scheduling policy for power-constrained transmission has been proposed in \cite{Infocom2003Goyal}.
In addition, system throughput maximization subject to a given constraint for low queueing delay is addressed in \cite{TWC2007Tang}.
Furthermore, the tradeoff between energy and delay to adapt to changes of network state distribution is discussed in~\cite{TIT2007Neely,Infocom2015Neely} based on a stochastic network optimization framework.

In IoT networks, it has been increasingly difficult for orthogonal multiple access (OMA), which allocates limited orthogonal resources to individual communication links, to handle the growing number of wireless devices.
In order to overcome this issue, non-orthogonal multiple access (NOMA) has been actively researched as one of promising methods for efficient and flexible use of energy and spectrum, as well as for system throughput improvements \cite{NOMA-5G:IEEE2017Liu, NOMA-5G:JSAC2017Ding}.
Power-multiplexing NOMA serves multiple users on the same time/frequency/code resources by employing the additional power domain \cite{NOMA_basic:Docomo-Saito}; thus it can provide better system throughput than OMA, employing successive interference cancellation (SIC) to remove superposed users' signal components \cite{Book:Tse}.
In addition, NOMA has the advantage of allowing massive connectivity for IoT services \cite{NOMA-IoT:CM2017Shirvan,NOMA-IoT:IEEE2016Ding} and NOMA in short packet transmissions for achieving low latency has been discussed in \cite{NOMA-shortpacket:IEEE2016Ding,NOMA-shortpacket:CL2018Yu}.

Since all users would not be served by NOMA due to high complexity of SIC, hybrid multiple access (MA), which allows the coexistence of NOMA and OMA, has been considered for next-generation communication systems. 
Representatively, multi-user superposition transmission (MUST) is adopted by the 3rd generation partnership project (3GPP) for 5G networks, which employs both power-domain NOMA and orthogonal frequency division multiple access (OFDMA) \cite{MuST:3GPP}.
For using hybrid MA, user scheduling and resource allocation are very critical issues. 
In \cite{NOMA-UserPairing:TVT2015Ding,NOMA-UserPairing:TC2017Liang,NOMA-UserPairing:WCL2018Kang}, user pairing schemes for NOMA signaling among multiple users have been studied.
The proposed schemes in \cite{NOMA-RM:TC2016Fang,NOMA-RM:TWC2016Di,NOMA-RM:TC2017Sun,NOMA-RM:JSAC2017Zhu} focus on joint optimization of sub-channel assignment and power allocation.
Further, cognitive-radio-inspired power control for NOMA is proposed in \cite{NOMA-UserPairing:TVT2015Ding,NOMA-PA:TWC2016Yang} to guarantee secondary users' QoS requirements.
However, power efficiency and low latency are not considered in \cite{NOMA-UserPairing:TVT2015Ding,NOMA-UserPairing:TC2017Liang,NOMA-UserPairing:WCL2018Kang,NOMA-RM:TWC2016Di,NOMA-RM:TC2017Sun,NOMA-RM:JSAC2017Zhu,NOMA-PA:TWC2016Yang}.
The authors of \cite{NOMA-RM:TC2016Fang} and \cite{NOMA-CrossLayer:JSAC2017Bao} proposed the power-efficient resource allocation policies, but they did not consider user pairing and latency problems. 

This paper proposes dynamic algorithms of joint user pairing and power control to maximize power-efficiency while achieving low latency as well as sufficient reliability in hybrid MA.
In particular, the long-term average data rate is considered as a user QoS requirement for sufficient reliability.
In addition, user scheduling and flexible use of resources are also captured in the proposed technique since the proposed optimization framework enables to determine whether the communication link is activated or not.
This paper shows that the proposed technique works well based on the short frame structure designed for URLLC.
The main contributions of the proposed technique can be summarized as follows:

\begin{itemize}
	\item This paper constructs the stochastic network optimization framework for the transmission scheme of URLLC, which adaptively operates depending on time-varying channel and queue states.
	The proposed framework focuses on reducing the queueing delay, which is a main factor of E2E delay.
	
	\item This paper contributes to URLLC based on NOMA. 
	The proposed transmission scheme exploits the advantage of NOMA over OMA to increase data rate for reducing the queueing delay. 
	Further, flexible use of resources is enabled for both OMA and NOMA users in the proposed framework.
	
	\item Different from the power-efficient methods of the existing works \cite{NOMA-RM:TC2016Fang,NOMA-CrossLayer:JSAC2017Bao}, the proposed resource allocation and user scheduling not only maximize the power efficiency, but also guarantee the limited queueing delay and the sufficiently large time-average data rate.
	
	\item Data-intensive simulation results show that the proposed technology can achieve an E2E delay smaller than 1\,ms, on the basis of the short packet structure designed for URLLC \cite{ShortPacket:VTC2015Kela}.
	
\end{itemize}

In summary, the proposed algorithm pursues low latency, high power-efficiency, as well as diverse QoS requirement satisfactions.

The rest of the paper is organized as follows. 
The hybrid MA system and queue model are described in Section \ref{sec:NetworkModel}.
In Section \ref{sec:joint_opt_prob}, we formulate the joint optimization problem of user pairing and power allocation in hybrid MA.
The optimal power allocation rule with fixed user pairing is proposed in Section \ref{sec:PA_NOMA}, and the matching algorithm for user pairing is presented in Section \ref{sec:user_pairing}.
Simulation results are shown in Section \ref{sec:simulation}, and Section \ref{sec:conclusion} concludes the paper.

\section{System Model}
\label{sec:NetworkModel}

\subsection{Hybrid Multiple Access Model}

This paper considers hybrid MA for power-efficient IoT networks where transmitters are battery-powered.
Let a transmitter serve $N$ users by either OMA or NOMA, as shown in Figs. \ref{Fig:System_Model_OMA} and \ref{Fig:System_Model_NOMA}.
The transmitter is deployed with $N$ queues in which data packets are waiting for transmissions to $N$ respective users.
Assume that data packets for user $n$ are accumulated in queue $n$.
Each transmitter queue has a power budget of $P_0$, and the transmit power for user $n$ is $P_n$, so $0\leq P_n \leq P_0$.

The Rayleigh fading channel is assumed for communication links from the transmitter to users. 
Denote the channel of user $n$ with $h_n$. 
The path loss model is $35.3+37.6 \ln(d_n)$, where $d_n$ is the distance between the transmitter and user $n$, and the fast fading component has a complex Gaussian distribution, i.e., $CN(0,1)$.
Let $R_n$ be the data rate of user $n$, and denote $\rho_n$ as a threshold of user $n$'s instantaneous data rate that determines the outage event.
In other words, when $R_n < \rho_n$, the outage event occurs at user $n$. 
In addition, $\eta_n$ represents the long-term average data rate as a QoS requirement for user $n$, and the QoS constraint can be written as
\begin{equation}
\eta_n \leq \lim_{T\rightarrow \infty} \frac{1}{T} \sum_{t=1}^{T} R_n(t).
\end{equation}

\begin{figure}[t]
	\minipage{0.45\textwidth}
	\includegraphics[width=\linewidth]{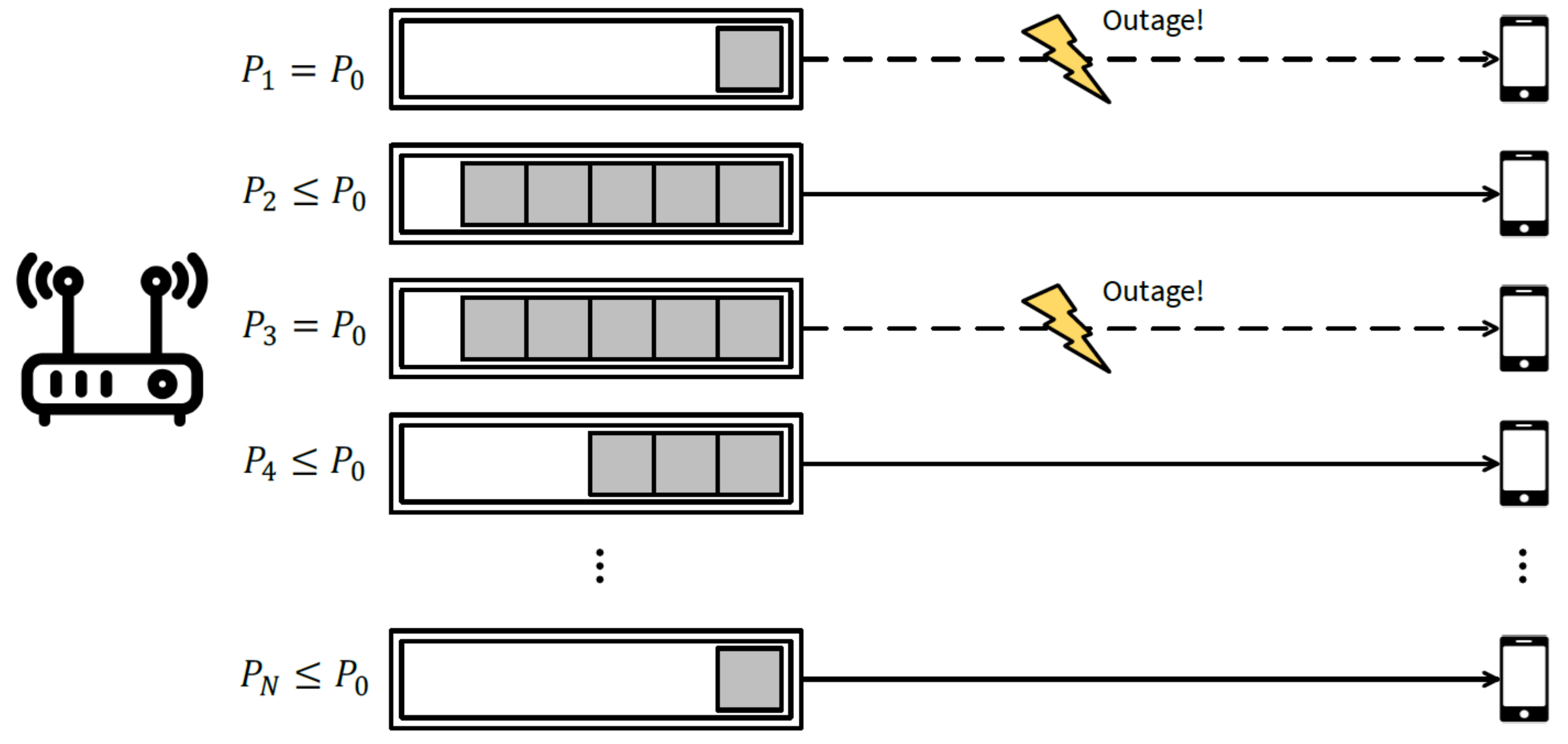}
	\caption{System architecture for OMA} 
	\label{Fig:System_Model_OMA}
	\endminipage\hfill
	\minipage{0.5\textwidth}
	\includegraphics[width=\linewidth]{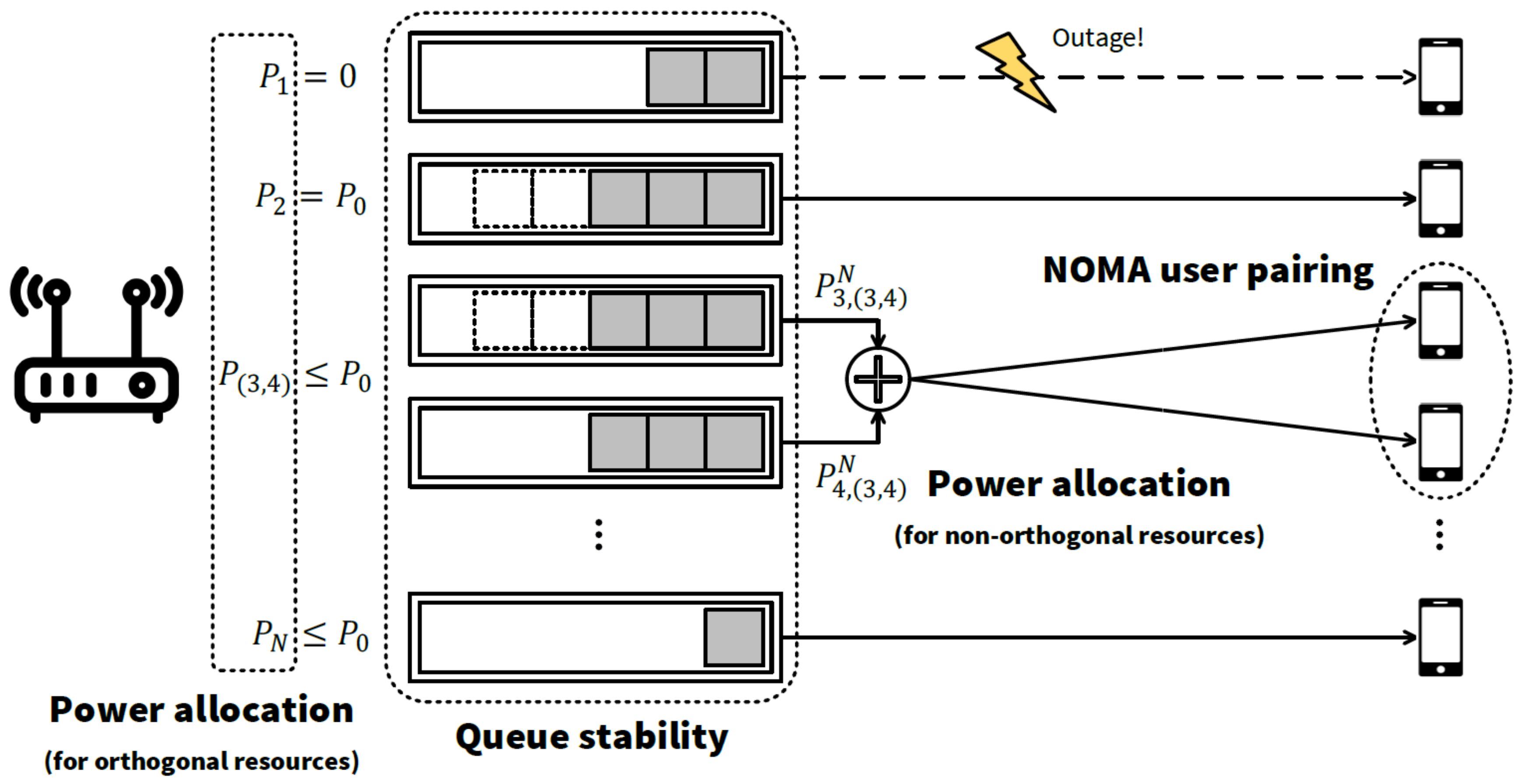}
	\caption{System architecture for hybrid MA}
	\label{Fig:System_Model_NOMA}
	\endminipage
\end{figure}

Fig. \ref{Fig:System_Model_OMA} shows the system architecture employing only OMA. 
For example, many packets are in queue 2 and queue 3, so these links have a risk of excessive queueing delays.
The transmitter can reduce queue backlogs by increasing transmission rates of these links.
Simply, link 2 can consume more power to increase its transmission rate and to reduce queue backlogs.
On the other hand, link 3 experiences the outage event even with the maximum transmit power in Fig. \ref{Fig:System_Model_OMA}, so NOMA can be employed.
Since NOMA is well-known to improve system throughput compared to OMA with identical power consumption, the link rate for user 3 can be increased by NOMA.
In Fig. \ref{Fig:System_Model_NOMA}, link 3 and link 4 are paired for NOMA transmission.
Meanwhile, the link outage occurs at user 1, but its queue is almost empty so it does not worry about the excessive queueing delay.
Since we suppose that the transmitter observes the current CSI and QSI, the outage occurrence can be expected.
In this case, the transmitter can allow link 1 to be deactivated and to save transmit power, as shown in Fig. \ref{Fig:System_Model_NOMA}. 
In this way, the system dynamically adjusts transmission rates of all communication links by controlling power consumption and employing NOMA, for low queueing delay and high power-efficiency.

Since this model determines link activation by allocating no transmit power, as shown by link 1 in Fig. \ref{Fig:System_Model_NOMA}, we can say that user scheduling is also performed in the system model.
Further, when two users are paired for NOMA signaling, the transmitter even enables to allocate no power to one of paired users.
In this case, the resource of the user with no transmit power becomes available for another user. 
It is easily seen that system resources are used more flexibly in this model.

Main issues of hybrid MA are summarized in Fig. \ref{Fig:System_Model_NOMA}.
First, the queueing delay should be reduced to satisfy the E2E latency constraint, by achieving stability of the queueing system, i.e., limiting queue backlogs.
Then, the user pairing problem arises for the transmitter to serve several users by NOMA.
In addition, power allocation for both OMA and NOMA users should be jointly considered with the user pairing problem.
In this paper, a two-user NOMA scenario is only considered, because a clumsy device in the IoT network is difficult to handle the high computational complexity of SIC processes for the multi-user NOMA scenario.
Further, as the number of users for NOMA signaling grows up, the power budget should be large enough to provide reliable signal-to-interference-plus-noise ratios (SINRs) to NOMA users. 
However, a small battery-powered device is likely to have a limited power budget.

\subsection{Transmitter Queue Model}
\label{subsec:queue_model}

In general, the transmitter queue model has its own arrival and departure processes. 
When the departures are less frequent than the arrivals, the queue backlog grows. 
For each user $n\in \{1,\cdots,N \}$, the queue dynamics in each unit time $t \in \{0,1,\cdots, \}$ can be represented as follows:
\begin{align}
Q_n[t+1] &= \max \{Q_n[t]-\mu_n[t], 0\} + \lambda_n[t] \\
Q_n[0] &= 0,
\end{align}
where $Q_n[t]$, $\lambda_n[t]$, and $\mu_n[t]$ stand for the queue backlog, the arrival and departure proesses of user $n$ at time $t$, respectively.
The queue states are updated in each time slot $t$. 
In this paper, the interval of each slot is assumed to be the channel coherence time, $\tau_c$. 

In this paper, queue backlog $Q_n[t]$ counts the number of data bits accumulated in queue $n$. 
$\lambda_n[t]$ and $\mu_n[t]$ semantically mean the numbers of arrived and transmitted bits.
Simply, suppose that $\lambda_n[t]$ is randomly generated for all $n\in \{1,\cdots,N\}$.
On the other hand, $\mu_n[t]$ obviously depends on the data rate of user $n$:

\begin{align}
\lambda_n[t] &= a_n[t] \cdot u \\
\mu_n[t] &= \mathcal{I}\{R_n(P_n,\Psi_n,t) \geq \rho_n \} \cdot R_n(P_n,\Psi_n,t) \cdot \tau_c
\end{align}
where $a_n[t]$ is an i.i.d. uniform random variable, i.e., $a_n[t] \sim \mathcal{U}\{\lambda_{\text{min}}, \lambda_{\text{max}} \}$, indicating the number of data packets arrived in queue $n$ at time $t$. 
Also, $u$ is the packet size in bits, and $\Psi_n$ represents the index of the user paired with user $n$. 
If OMA is employed for user $n$, then $\Psi_n = n$, whereas $\Psi_n = m$ for $m \neq n$ means that users $n$ and $m$ are paired for NOMA.
$R_n(P_n,\Psi_n,t)$ is the data rate of user $n$ when transmit power $P_n$ is consumed and user $n$ is paired with user $\Psi_n$ at time $t$. 
$\mathcal{I}(.)$ is the indicator function, so $\mathcal{I}\{R_n(P_n,\Psi_n,t) \geq \rho_n \}$ is 0 if the outage occurs at user $n$, or 1, otherwise.
Since we suppose that the transmitter can observe the current CSI, if the outage is expected, no transmit power is allocated and the departure becomes zero.

\textit{Remark}: If $\tau_c$ is too long, it is better to update power allocation and user pairing more frequently than channel variations. 
Consider a transmitter queue which is almost empty so that there is no worry about excessive queueing delay.
In this case, the transmitter usually consumes a small power to improve power efficiency.
However, if this situation persists for a long time $\tau_c$, packets will be accumulated in the queue sooner or later and queueing delay will increase. 
Therefore, several updates of power allocation and user pairing are required over the time interval of $\tau_c$.

\subsection{E2E Delay Requirement}
\label{subsec:E2E_delay}

Denote the E2E delay bound with $D_{\text{max}}$. 
$D_{\text{max}}$ mainly consists of UL/DL transmission delays and the queueing delay \cite{QueueDelay:VTC2016She}.
For low latency communications, a small packet structure is preferred because the UL/DL transmission durations can be reduced, and the summation of UL/DL durations becomes identical to the transmit time interval (TTI), denoted by $T_t$ \cite{ShortPacket:VTC2015Kela}.
Therefore, the margin of the queueing delay is $D_{\text{max}}^q = D_{\text{max}} - T_t$, i.e., data transmission is successful only when the queueing delay is smaller than $D_{\text{max}}^q$.
Although the instantaneous queueing delay determines whether data transmission is successful or not, making the instantaneous queueing delay bounded to a deterministic value is very difficult due to the time-varying transmission scheme and channel.

To this end, this paper focuses on limiting the time-average queueing delay.
According to the Little's theorem \cite{Book:DataNetworks}, low time-average queueing delay can be achieved by reducing the expected value of queue backlogs. 
This paper introduces the concept of strong stability of a queue to make queue backlogs bounded, as follows:
\begin{equation}
\lim_{T \rightarrow \infty } \frac{1}{T} \sum_{t=0}^{T} Q[t] < \infty.
\label{eq:queue_stability}
\end{equation}
Simulation results in section \ref{sec:simulation} show that the queueing delay can be reduced by ensuring (\ref{eq:queue_stability}), i.e., strong stability of the queueing system.

\section{Joint Optimization Problem Formulation of User pairing and Power Allocation in Hybrid Multiple Access}
\label{sec:joint_opt_prob}

This paper pursues both high power efficiency and low queueing delay. 
In addition, the long-term average data rates are considered as one of the QoS requirements.
Specifically, the transmit power of $P_n$ depends on whether the transmitter serves user $n$ by OMA or NOMA and which user is paired with user $n$ for NOMA. 
Thus, we can formulate the joint optimization problem to find the optimal power allocation and user pairing:
\begin{align}
\{ \mathbf{P}^*(\boldsymbol{\Psi}^*(t),t), \boldsymbol{\Psi}^*(t) \} = ~&\underset{\mathbf{P},\boldsymbol{\Psi}}{\arg\min} ~\sum_{n\in \mathcal{N}} \mathbb{E}[P_n(\Psi_n(t),t)] \label{eq:OP1} \\
\text{s.t.}~~ &\lim_{t \rightarrow \infty} \frac{1}{t}\sum_{t'=0}^t \mathbb{E}[Q_n(t')] < \infty,~\forall n \in \mathcal{N} \label{eq:OP1_QueueConst} \\
&\lim_{t \rightarrow \infty} \frac{1}{t}\sum_{t'=0}^{t-1} \mathbb{E} [\tilde{R}_l(P_l(\Psi_l(t'),t'),\Psi_l(t'),t')] \geq \eta_l,~\forall l \in \mathcal{N}_{s} \subseteq \mathcal{N} \label{eq:OP1_RateConst} \\
&0 \leq P_n(\Psi_n(t),t) \leq P_0,~\forall n \in \mathcal{N}, \label{eq:OP1_Powerconst}
\end{align}
where $\tilde{R}_n(.) = \mathcal{I}\{R_n(.) \geq \rho_n \} \cdot R_n(.)$,
$\mathcal{N}=\{1,\cdots,N \}$ is the user index set, and $\mathcal{N}_{s}$ is the subset of $\mathcal{N}$.
$\mathbf{P}^*(\boldsymbol{\Psi}^*(t),t)$ and $\boldsymbol{\Psi}^*(t)$ denote the column vectors of $P_n^*(\Psi^*_n(t),t)$ and $\Psi_n^*(t)$ for all $n\in \mathcal{N}$, respectively. 
Note that the transmit power $P_n$ depends on user pairing $\Psi_n$ and time $t$.
The constraint (\ref{eq:OP1_QueueConst}) represents strong stability of the queueing system, which makes queue backlogs bounded.
In addition, sufficient time-average data rates of $\eta_l$ for user $l$ for $l\in \mathcal{N}_s$ are guaranteed as one of QoS by the constraint (\ref{eq:OP1_RateConst}).
The power budget of $P_0$ is independently assumed for each communication link, so the constraint (\ref{eq:OP1_Powerconst}) is given.
For simplicity, $P_n(\Psi_n(t),t)$ will be written simply as $P_n(t)$.

The problem (\ref{eq:OP1})-(\ref{eq:OP1_Powerconst}) can be solved by the theory of Lyapunov optimization \cite{Book:Neely}. 
We first transform the inequality constraint (\ref{eq:OP1_RateConst}) into the form of queue stability.
Specifically, define the virtual queue $Z_l(t)$ for $l\in \mathcal{N}_s$, with the update equation:
\begin{equation}
Z_l(t+1) = \max \{ Z_l(t) + \eta_l - \tilde{R}_l(P_l(t), \Psi_l (t), t), 0 \}.
\end{equation}
The strong stability of the virtual queue $Z_l(t)$ pushes the average of $\tilde{R}_l(P_l(t),\Psi_l(t),t)$ to be close to the QoS guarantee $\eta_l$.

Let $\mathbf{Q}(t)$ and $\mathbf{Z}(t)$ denote the column vectors of $Q_n(t)$ and $Z_l(t)$ for $n\in \mathcal{N}$ and $l \in \mathcal{N}_s$ at time $t$, respectively, and let $\mathbf{\Theta}(t) = [\mathbf{Q}(t)^T, \mathbf{Z}(t)^T]^T$ be a concatenated vector of actual and virtual queue backlogs.
Define the quadratic Lyapunov function $L[\mathbf{\Theta}(t)]$ as follows:
\begin{equation}
L[\mathbf{\Theta}(t)] = \frac{1}{2} \sum_{n\in \mathcal{N}} Q_n(t)^2 + \frac{1}{2} \sum_{l \in \mathcal{N}_s} Z_l(t)^2.
\end{equation}
Then, let $\Delta (t)$ be a conditional quadratic Lyapunov function that can be formulated as $\mathbb{E}[L[\mathbf{\Theta}(t+1)] - L[\mathbf{\Theta}(t)]|\mathbf{\Theta}(t)]$, i.e., the drift on $t$. 
The dynamic policy is designed to solve the given optimization problem (\ref{eq:OP1})-(\ref{eq:OP1_Powerconst}) by observing the current queue state, $\boldsymbol{\Theta}(t)$, and determining power allocation $\mathbf{P}(t)$ and user pairing $\boldsymbol{\Psi}(t)$ to minimize a upper bound on \textit{drift-plus-penalty} \cite{Book:Neely}:
\begin{equation}
\Delta(\mathbf{\Theta}(t)) + V \mathbb{E}\bigg[ \sum_{n\in \mathcal{N}} P_n(t) \Big| \mathbf{\Theta}(t) \bigg].
\label{eq:drift-plus-penalty}
\end{equation}

First, find the upper bound on the change in the Lyapunov function.
\begin{align}
L[\boldsymbol{\Theta}(t+1)] - L[\boldsymbol{\Theta}(t)] &= \frac{1}{2} \sum_{n\in \mathcal{N}} \Big[ Q_n(t+1)^2 - Q_n(t)^2 \Big] + \frac{1}{2} \sum_{l\in \mathcal{N}_s} \Big[ Z_l(t+1)^2 - Z_l(t)^2 \Big] \\
&\leq \frac{1}{2} \sum_{n\in \mathcal{N}} [ \lambda_n(t)^2 + \mu_n(t)^2 ] + \frac{1}{2} \sum_{l \in \mathcal{N}_s} (\eta_l - \tilde{R}_l(P_l(t),\Psi_l(t),t))^2 \nonumber \\ 
&~~~+ \sum_{n\in \mathcal{N}} Q_n(t) (\lambda_n(t) - \mu_n(t))
+ \sum_{l\in \mathcal{N}_s} Z_l(t) (\eta_l - \tilde{R}_l (P_l(t),\Psi_l(t),t)).
\end{align}
Then, the upper bound on the conditional Lyapunov drift is given by
\begin{align}
\Delta(\Theta(t)) &\leq C + \sum_{n\in \mathcal{N}} \mathbb{E} \Big[ Q_n(t) (\lambda_n(t) - \mu_n(t)) \Big]
+ \sum_{l\in \mathcal{N}_s} \Big[ Z_l(t) (\eta_l - \tilde{R}_l (P_l(t),\Psi_l(t),t)) \Big].
\end{align}
where we assume that departure and arrival rates are bounded, and $C$ is a constant such that $\frac{1}{2} \sum_{n\in \mathcal{N}} \mathbb{E} [ \lambda_n(t)^2 + \mu_n(t)^2 ] + \frac{1}{2} \sum_{l \in \mathcal{N}_s} \mathbb{E}[ (\eta_l - \tilde{R}_l(P_l(t),\Psi_l(t),t))^2 ] \leq C$.
According to (\ref{eq:drift-plus-penalty}), minimizing a bound on \textit{drift-plus-penalty} is consistent with minimizing 
\begin{equation}
\mathbb{E}\Bigg[V\sum_{n \in \mathcal{N}} P_n(t) - \sum_{n\in \mathcal{N}} Q_n(t)\mu_n(t) - \sum_{\l \in \mathcal{N}_s} Z_l(t) \tilde{R}_l(P_l(t),\Psi_l(t), t) \Big| \boldsymbol{\Theta}(t) \Bigg],
\label{eq:expected-drift-plus-penalty}
\end{equation}
because $\lambda_n(t)$ is not controllable and all values of $\eta_l$ for $l\in \mathcal{N}_s$ are constants.

We now use the concept of opportunistically minimizing the expectations and specifically go after the following drift-plus-penalty problem:
\begin{align}
\{ \mathbf{P}^*(t), \boldsymbol{\Psi}^*(t) \} = ~&\underset{\mathbf{P},\boldsymbol{\Psi}}{\arg\min} ~\mathcal{M}(\mathbf{P}(t),\boldsymbol{\Psi}(t)) \label{eq:OP2} \\
\text{s.t.}~~&0 \leq P_n(t) \leq P_0,~\forall n \in \mathcal{N}, \label{eq:OP2_const}
\end{align}
where 
\begin{align}
\mathcal{M}(\mathbf{P}(t),\boldsymbol{\Psi}(t)) &= \sum_{n \in \mathcal{N}} \mathcal{M}_n(P_n(t),\Psi_n(t)) \\
&= V\sum_{n \in \mathcal{N}} P_n(t) - \sum_{n\in \mathcal{N}} Q_n(t)\mu_n(t) - \sum_{\l \in \mathcal{N}_s} Z_l(t) \tilde{R}_l(P_l(t),\Psi_l(t), t).
\label{eq:Lyapunov_metric}
\end{align}

Since there are so many possible combinations of user pairing, it is very difficult to exhaustively minimize the optimization metric of (\ref{eq:Lyapunov_metric}).
Therefore, we first find the optimal power allocation depending on the fixed user pairing policy. 
Then, several pairs of two users are generated for NOMA to minimize the optimization metric of (\ref{eq:Lyapunov_metric}), based on the matching theory.

\section{Optimal Power Allocation for Hybrid MA}

For simplicity, notations for the dependency of all parameters on $t$ are omitted in this section, because the optimal power allocation depends only on CSI and QSI at current time $t$.
Therefore, $R_n(P_n(t),\Psi_n(t),t) = R_n(P_n,\Psi_n)$, in this section.

\subsection{Optimal power allocation for OMA}

First, the power allocation policy for OMA users is presented.
The data rate of user $n$ which employs OMA is given by
\begin{equation}
R_n(P_n,n) = \frac{\Phi \mathcal{B}}{N} \log_2 \Big( 1+ N \Gamma_n P_n \Big),
\label{eq:OMA_rate}
\end{equation}
where $\Gamma_n = \frac{|h_n|^2}{\mathcal{B}N_0}$, $N_0$ is single-sided noise spectral density, and $\mathcal{B}$ is bandwidth. 
$\Phi \in (0,1]$ represents the degradation coefficient of channel capacity due to finite blocklength codes appropriate for the short packet structure \cite{QueueDelay:VTC2016She,finiteblock:TIT2010Polyanskiy}. 
We assume that the bandwidth is equally allocated to $N$ users.
Note that $\Psi_n = n$ for all $n \in \mathcal{N}$ in OMA.
The power interval for avoiding the outage, i.e. $R_n(P_n,n) \geq \rho_n$, can be obtained, as written by
\begin{equation}
P_n \geq P^O_{\text{th}} = \frac{2^{N\rho_n/\Phi\mathcal{B}}-1}{N \Gamma_n}. \label{eq:OMA_outage_power}
\end{equation}
If $P_n \leq P^O_{\text{th}}$, then $R_n(P_n,n) = 0$.

\textit{Remark}: Since Shannon capacity assumes channel codes of infinite length, it is not appropriate to directly apply Shannon capacity to low latency communications with the short packet structure.
The authors of \cite{finiteblock:TIT2010Polyanskiy,finiteblock:TIT2014Yang} obtained channel capacity with finite blocklength codes in a variety of channel models.
Strictly, the capacity with finite blocklength codes has a different form from (\ref{eq:OMA_rate}); 
however, we approximate the data rate by weighting the degradation factor to Shannon capacity in a similar way to that in \cite{QueueDelay:VTC2016She}, and $\Phi=0.9$ is assumed in this paper.

When the transmitter serves all of $N$ users by OMA, each user's data rate $R_n(P_n,n)$ is independent of each other, so the optimization problem (\ref{eq:OP2})-(\ref{eq:OP2_const}) can be solved by independently minimizing $\mathcal{M}_n(P_n,\Psi_n)$ for all $n \in \mathcal{N}$.
When $\Psi_n = n$, let $\mathcal{M}^O_n(P_n) = \mathcal{M}_n(P_n,\Psi_n)$, $\forall n \in \mathcal{N}$.
Therefore, the optimization problem (\ref{eq:OP2})-(\ref{eq:OP2_const}) can be transformed into 
\begin{align}
P^O_n = ~&\underset{P_n}{\arg\min}~\mathcal{M}^O_n(P_n) \label{eq:OP_PA_OMA} \\
\text{s.t.}~~&0 \leq P_n \leq P_0, \label{eq:OP_PA_OMA_const}
\end{align}
for all $n\in\mathcal{N}$ in OMA system, where 
\begin{align}
\mathcal{M}^O_n(P_n) &= V \cdot P_n - Q_n \cdot \tilde{R}_n(P_n,n) \cdot \tau - Z_n \cdot \tilde{R}_n(P_n,n) \cdot \mathcal{I} \{ n \in \mathcal{N}_s \}.
\label{eq:OP_PA_OMA_metric}
\end{align}

Theorem 1 provides the solution of the optimization problem (\ref{eq:OP_PA_OMA})-(\ref{eq:OP_PA_OMA_const}).
\begin{theorem}
	
	\label{thm1}
	The optimal power allocation of the problem (\ref{eq:OP_PA_OMA})-(\ref{eq:OP_PA_OMA_const}) is given by
	\begin{equation}
		P_n^O = 
		\begin{cases}
		P_0,~~& \text{if}~ P^O_{\text{th}} \leq P_0 \leq P_n^* ~\&~ \mathcal{M}_n(P_0) < 0 \\
		P_n^*,~~&\text{else if}~ P^O_{\text{th}} \leq P_n^* < P_0 \\
		P_{\text{th}}^O,~~&\text{else if}~ P_n^* < P_{\text{th}}^O \leq P_0 ~\&~ \mathcal{M}_n (P_{\text{th}}^O) < 0 \\
		0,~~&\text{otherwise}
		\end{cases} \label{eq:thm1},
	\end{equation}
	where $P_n^* = \frac{\Phi \mathcal{B}(\tau Q_n + Z_n \cdot \mathcal{I}\{ n \in \mathcal{N}_s \})}{NV\ln 2} - \frac{1}{N \Gamma_n}$.
	
\end{theorem}

\begin{proof}
	
	Assume that $P_n \geq P_{\text{th}}^O$, i.e., the outage event does not occur.
	Then, differentiating (\ref{eq:OP_PA_OMA_metric}) by $P_n$,
	\begin{equation}
	\frac{\mathrm d \mathcal M_n^O}{\mathrm d P_n} = V - \frac{\Phi\mathcal{B}(\tau Q_n + \tilde{Z}_n)}{\ln 2} \cdot \frac{\Gamma_n}{N \Gamma_n P_n + 1},
	\end{equation}
	where $\tilde{Z}_n = Z_n \cdot \mathcal{I}\{ n \in \mathcal{N}_s \}$.
	and
	the local minimizer $P_n^*$ is obtained from $\frac{\mathrm d \mathcal M_n^O}{\mathrm d P_n}=0$, i.e.,
	\begin{equation}
	P_n^* = \frac{\Phi\mathcal{B}(\tau Q_n + \tilde{Z}_n)}{NV\ln 2} - \frac{1}{N \Gamma_n}.
	\end{equation}
	Further, $P_n^*$ is shown to be the global minimizer in the region of $P_n \geq P_{\text{th}}^O$ by
	\begin{equation}
	\frac{\mathrm d^2 \mathcal M_n^O}{\mathrm d P_n^2} = \frac{N\Gamma_n^2 \Phi \mathcal{B} (\tau Q_n + \tilde{Z}_n)}{(N\Gamma_n P_n + 1)^2 \ln 2} > 0.
	\end{equation}
	However, when $P_n < P_{\text{th}}^O$, $\mathcal{M}_n^O(P_n) = V\cdot P_n$, so $P_n=0$ is the minimizer and $\mathcal{M}_n^O(P_n) = 0$.
	Therefore, if $P_{\text{th}}^O > P_0$, $P_n^O = 0$ always.
	Otherwise, i.e. when $P_{\text{th}}^O \leq P_0$, the relative value of $P_n^*$ to $P_{\text{th}}^O$ and $P_0$ determines $P_n^O$.
	
	When $P_{\text{th}}^O \leq P_n^* < P_0$, $P_n^*$ is still the global minimizer.
	However, when $P_0 \leq P_n^*$, the minimizer in the interval of $[P_{th}^O,P_0]$ becomes $P_0$, and still $P_n=0$ is the minimizer in $[0, P_{th}^O]$. 
	Therefore, if $\mathcal{M}_n^O(P_0) < 0$, $P_0$ is the global minimizer in $[0,P_0]$.
	Otherwise, $P_n^O=0$, i.e., no power is allocated to user $n$.
	
	In the case of $P_n^* < P_{\text{th}}^O$, the global minimizer is in the outage region.
	Then, $P_{\text{th}}^O$ is the minimizer in the interval of $[P_{\text{th}}^O,P_0]$.
	Thus, if $\mathcal{M}_n^O(P_{\text{th}}^O) < 0$, $P_n^O = P_{\text{th}}^O$ becomes the minimizer in $[0,P_0]$, and if not, $P_n^O=n$ is the solution.
	Finally, (\ref{eq:thm1}) is obtained.
	
\end{proof}

\textit{Remark}: 
As we mentioned earlier, when the outage is expected at user $n$, the transmitter can save the power, i.e., $P_n = 0$.
Further, when queue backlogs of $Q_n$ and $Z_n$ are small so the second and third terms of (\ref{eq:OP_PA_OMA_metric}) are small compared to the system parameter $V$, the transmitter cannot schedule the link of user $n$ to save the power, even though the link is not in outage. 
In this way, link activation is determined by power allocation, so it can be said that user scheduling is also performed.

\subsection{Optimal Power Allocation for Two-user NOMA}
\label{sec:PA_NOMA}

In this section, the optimal power allocation is obtained for a given NOMA pair of user $i$ and user $j$, i.e., $\Psi_i = j$ and $\Psi_j = i$.
Assume that $|h_j|^2 > |h_i|^2$.
For employing NOMA, the larger power is usually allocated to the user with weaker channel condition. 
Throughout the paper, the user with the weaker channel who does not perform SIC and the user with the stronger channel who performs SIC will be referred to as non-SIC user and SIC user, respectively.
Let user $i$ and user $j$ be the non-SIC user and the SIC user, respectively, without loss of generality, with the assumption of $P_i \geq P_j$.
The data rates of NOMA users are given by
\begin{align}
R_i(P_i,\Psi_i = j) &= \frac{2\Phi\mathcal{B}}{N} \log_2 \bigg( 1 + \frac{N \Gamma_i P_i/2}{N  \Gamma_i P_j /2 + 1} \bigg) \\
R_j(P_j,\Psi_j = i) &= \frac{2\Phi\mathcal{B}}{N} \log_2 ( 1 + N \Gamma_j P_j / 2 ).
\end{align}

Suppose that signals for other users are orthogonally multiplexed with the NOMA signaling of user $i$ and user $j$.
Then, the power allocation problem for user $i$ and user $j$ can be independently formulated from the power allocation problem of (\ref{eq:OP2})-(\ref{eq:OP2_const}), as follows:
\begin{align}
\{P^N_{i,(i,j)}, P^N_{j,(i,j)} \} = ~&\underset{P_i,P_j}{\arg\min}~\mathcal{M}_{(i,j)}^N(P_i, P_j) = \mathcal{M}^N_{i,(i,j)} (P_i) + \mathcal{M}^N_{j,(i,j)} (P_j) \label{eq:OP_PA_NOMA}\\
\text{s.t.}~~ &0\leq P_i, P_j \leq P_0, \label{eq:OP_PA_NOMA_const}
\end{align}
where $\mathcal{M}_{i,(i,j)}^N(P_i) = V \cdot P_i - (\tau Q_i + \tilde{Z}_i) \cdot \tilde{R}_i(P_i,\Psi_i = j)$, which is the optimization metric of user $i$ when user $i$ is paired with user $j$ and $j \neq i$. 
$P_{j,(i,j)}^N$ represents the optimal transmit power for user $j$, when user $j$ is paired with user $i$ for NOMA.
However, $\mathcal{M}_{(i,j)}^N(P_i, P_j)$ is not concave, so the optimization problem of (\ref{eq:OP_PA_NOMA})-(\ref{eq:OP_PA_NOMA_const}) is not a convex problem.
Therefore, the auxiliary variable $q=P_i+P_j$ is introduced. 
$q$ represents power summation of a user pair, so $q \leq 2P_0$ should be satisfied.
Then, the problem of (\ref{eq:OP_PA_NOMA})-(\ref{eq:OP_PA_NOMA_const}) can be resolved by solving two sequential subproblems.

The first subproblem is to find the power allocation for NOMA users with the fixed value of $q$, as formulated by:
\begin{align}
P_{j,(i,j)}^N = ~&\underset{P_j}{\arg \min} ~g(P_j) \label{eq:OP_PA_NOMA1}\\
\text{s.t.}~~ &0\leq P_j \leq P_i \leq P_0, \label{eq:OP_PA_NOMA1_const1} \\
&P_i + P_j = q, \label{eq:OP_PA_NOMA1_const2}
\end{align}
where 
\begin{equation}
g(P_j) = V\cdot q - \frac{2\Phi\mathcal{B}(\tau Q_i + \tilde{Z}_i)}{N} \log_2 \Big( \frac{N \Gamma_i q / 2 + 1}{N \Gamma_i P_j / 2 + 1} \Big) - \frac{2\Phi\mathcal{B}(\tau Q_j + \tilde{Z}_j)}{N} \log_2 ( 1 + N \Gamma_j P_j / 2 ).
\label{eq:OP_PA_NOMA1_metric}
\end{equation}

The power intervals for avoiding the outage event at both NOMA users is considered to solve the subproblem of (\ref{eq:OP_PA_NOMA1})-(\ref{eq:OP_PA_NOMA1_const2}).
$R_j(P_j,i) \geq \rho_j$ should be guaranteed for user $j$ to avoid the outage, in other words, the transmit power should be 
\begin{equation}
P_j \geq P_{j,(i,j)}^o = \frac{2}{N \Gamma_j}(2^{N\rho_j/2\Phi\mathcal{B}}-1). 
\label{eq:P_{j,(i,j)}^o}
\end{equation}
Let $\mathcal{O}_j = [0,P_{j,(i,j)}^o]$ denote the outage region of user $j$.
When $P_j \in \mathcal{O}_j$, the objective function of (\ref{eq:OP_PA_NOMA1_metric}) is given by
\begin{equation}
g_{\mathcal{O}_j}(P_j) = V\cdot q - \frac{2\Phi\mathcal{B}(\tau Q_i + \tilde{Z}_i)}{N} \log_2 \Big( \frac{N\Gamma_i q /2 + 1}{N\Gamma_i P_j / 2 + 1} \Big).
\end{equation}
Similarly, user $i$ can prevent the outage event when $R_i^N(P_i,j) < \rho_i$, and it corresponds to 
\begin{equation}
P_j \leq P_{i,(i,j)}^o = \Big[q-\frac{2}{N \Gamma_i}2^{N\rho_i/2\Phi\mathcal{B}} \Big] \cdot 2^{-N\rho_i/2\Phi\mathcal{B}}.
\label{eq:P_{i,(i,j)}^o}
\end{equation}
Let $\mathcal{O}_i = [P_{i,(i,j)}^o, q]$ denote the outage region of user $i$.
When $P_j \in \mathcal{O}_i$, the objective function of (\ref{eq:OP_PA_NOMA1_metric}) is given by
\begin{equation}
g_{\mathcal{O}_i}(P_j) = V\cdot q - \frac{2\Phi\mathcal{B}(\tau Q_j + \tilde{Z}_j)}{N} \log_2 \Big( 1 + N\Gamma_j P_j /2 \Big).
\end{equation}

Since $N\rho_i / 2\Phi\mathcal{B} > 0$ always, $0 \leq P_{j,(i,j)}^o$ and $P_{i,(i,j)}^o \leq q$ are guaranteed.
Then, Theorem \ref{thm2} gives the solution to the subproblem of (\ref{eq:OP_PA_NOMA1})-(\ref{eq:OP_PA_NOMA1_const2}). 

\begin{theorem}
	\label{thm2}
	Suppose that $1 < \frac{\tau Q_i + \tilde{Z}_i }{\tau Q_j + \tilde{Z}_j } < \frac{\Gamma_j}{\Gamma_i}$. 
	When $P_{j,(i,j)}^o \leq P_{i,(i,j)}^o$, the optimal power allocation of the problem (\ref{eq:OP_PA_NOMA1})-(\ref{eq:OP_PA_NOMA1_const2}) is given by (\ref{eq:thm2_sol1})-(\ref{eq:thm2_sol6}) unless $g(P_{j,(i,j)}^N)<0$, where $\bar{q} = \max(0,q-P_0)$ and 
	\begin{equation}
	P_j^* = \frac{2}{N \Gamma_i \Gamma_j} \cdot \frac{ \Gamma_j (\tau Q_j + \tilde{Z}_j) - \Gamma_i (\tau Q_i + \tilde{Z}_i) }{\tau Q_j + \tilde{Z}_j - \tau Q_i - \tilde{Z}_i }. \label{eq:OP_PA_NOMA1_P_j^*}
	\end{equation}
	When, $P_{j,(i,j)}^o > P_{i,(i,j)}^o$, the optimal power allocation is given by (\ref{eq:thm2_sol7})-(\ref{eq:thm2_sol9}) unless $g(P_{j,(i,j)}^N)<0$.
	If $g(P_{j,(i,j)}^N) \geq 0$, NOMA becomes useless for user $i$ and user $j$.
	
	\begin{itemize}
		\item When $P_{i,(i,j)}^o \leq \frac{q}{2}$ and $q-P_0 \leq P_{j,(i,j)}^o$, let $g_{\text{min}}(x) = \min \{ g(x), g_{\mathcal{O}_i}(q/2), g_{\mathcal{O}_j}(\bar{q}) \}$, then
		\begin{equation}
		P_{j,(i,j)}^N = 
		\begin{cases}
		x, & \text{if } g_{\text{min}}(x) = g(x) \\
		q/2, & \text{if } g_{\text{min}}(x) = g_{\mathcal{O}_i}(q/2) \\
		\bar{q}, & \text{if } g_{\text{min}}(x) = g_{\mathcal{O}_j}(\bar{q})
		\end{cases},~~
		\text{where}~
		x = 
		\begin{cases}
		P_{j,(i,j)}^o, & \text{if } \bar{q} < P_j^* < P_{j,(i,j)}^o \\
		P_j^*, & \text{if } P_{j,(i,j)}^o < P_j^* < P_{i,(i,j)}^o \\
		P_{i,(i,j)}^o, & \text{if } P_{i,(i,j)}^o < P_j^*
		\end{cases}
		\label{eq:thm2_sol1}
		\end{equation}
		
		\item When $P_{j,(i,j)}^o \leq \frac{q}{2} \leq P_{i,(i,j)}^o$ and $q - P_0 \leq P_{j,(i,j)}^o$, let $g_{\text{min}}(x) = \min \{ g(x), g_{\mathcal{O}_j}(\bar{q}) \}$, then
		\begin{equation}
		P_{j,(i,j)}^N = 
		\begin{cases}
		x, & \text{if } g_{\text{min}}(x) = g(x) \\
		\bar{q}, & \text{if } g_{\text{min}}(x) = g_{\mathcal{O}_j}(\bar{q})
		\end{cases},~~
		\text{where}~
		x = 
		\begin{cases}
		P_{j,(i,j)}^o, & \text{if } \bar{q} < P_j^* < P_{j,(i,j)}^o \\
		P_j^*, & \text{if } P_{j,(i,j)}^o < P_j^* < q/2 \\
		q/2, & \text{if } q/2 < P_j^*
		\end{cases}
		\label{eq:thm2_sol2}
		\end{equation}
		
		\item When $P_{i,(i,j)}^o \leq \frac{q}{2}$ and $P_{j,(i,j)}^o \leq q - P_0 \leq P_{i,(i,j)}^o$, let $g_{\text{min}}(x) = \min \{ g(x), g_{\mathcal{O}_i}(q/2) \}$, then
		\begin{equation}
		P_{j,(i,j)}^N = 
		\begin{cases}
		x, & \text{if } g_{\text{min}}(x) = g(x) \\
		q/2, & \text{if } g_{\text{min}}(x) = g_{\mathcal{O}_i}(q/2)
		\end{cases},~~
		\text{where}~
		x = 
		\begin{cases}
		q - P_0, & \text{if } P_j^* < q - P_0 \\
		P_j^*, & \text{if } q - P_0 < P_j^* < P_{i,(i,j)}^o \\
		P_{i,(i,j)}^o, & \text{if } P_{i,(i,j)}^o < P_j^*
		\end{cases}
		\label{eq:thm2_sol3}
		\end{equation}
		
		\item When $P_{j,(i,j)}^o \leq q - P_0 \leq q/2 \leq P_{i,(i,j)}^o$, then 
		\begin{equation}
		P_{j,(i,j)}^N = 
		\begin{cases}
		q-P_0, & \text{if } P_j^* < q - P_0 \\
		P_j^*, & \text{if } q-P_0 < P_j^* < P_{i,(i,j)}^o \\
		P_{i,(i,j)}^o, & \text{if } P_{i,(i,j)}^o < P_j^*
		\end{cases}
		\label{eq:thm2_sol4}
		\end{equation}
		
		\item When $P_{i,(i,j)}^o \leq q-P_0$, 
		\begin{equation}
		P_{j,(i,j)}^N = q/2 \label{eq:thm2_sol5}
		\end{equation}
		
		\item When $\frac{q}{2} \leq P_{j,(i,j)}^o$, 
		\begin{equation}
		P_{j,(i,j)}^N = \bar{q} \label{eq:thm2_sol6}
		\end{equation}
		
		\item When $\bar{q} \leq P_{i,(i,j)}^o < P_{j,(i,j)}^o < q/2$, let $g_{\text{min}} = \min\{ g_{\mathcal{O}_j}(\bar{q}), g_{\mathcal{O}_i}(\frac{q}{2}) \}$, then
		\begin{equation}
		P_{j,(i,j)}^N = 
		\begin{cases}
		P_{i,(i,j)}^o, & \text{if } g_{\text{min}} = g_{\mathcal{O}_j}(P_{i,(i,j)}^o) \\
		P_{j,(i,j)}^o, & \text{if } g_{\text{min}} = g_{\mathcal{O}_i}(P_{j,(i,j)}^o) 
		\end{cases}
		\label{eq:thm2_sol7}
		\end{equation}
		
		\item When $\bar{q} \leq P_{i,(i,j)}^o$ and $\frac{q}{2} \leq P_{j,(i,j)}^o$, then
		\begin{equation}
		P_{j,(i,j)}^N = \bar{q} \label{eq:thm2_sol8}
		\end{equation}
		
		\item When $P_{i,(i,j)}^o < \bar{q}$ and $P_{j,(i,j)}^o < \frac{q}{2}$, then
		\begin{equation}
		P_{j,(i,j)}^N = \frac{q}{2} \label{eq:thm2_sol9}
		\end{equation}
		
	\end{itemize}
	
\end{theorem}

\begin{proof}
	
	The constraints (\ref{eq:OP_PA_NOMA1_const1}) and (\ref{eq:OP_PA_NOMA1_const2}) can be combined together, as written by
	\begin{equation}
	\bar{q} \leq P_j \leq q/2 (\leq P_0). \label{eq:OP_PA_NOMA1_PowerIntrval}
	\end{equation}
	Differentiating the objective function of $g(P_j)$ by $P_j$, 
	\begin{equation}
	\frac{\mathrm d g}{\mathrm d P_j} = \frac{\Gamma_i (\tau Q_i + \tilde{Z}_i)}{N \Gamma_i P_j/2 + 1} - \frac{\Gamma_j (\tau Q_j + \tilde{Z}_j)}{N \Gamma_j P_j/2 + 1}.
	\end{equation}
	The local minimizer $P_j^*$ is obtained from $\frac{\mathrm d g}{\mathrm d P_j} = 0$, as given by (\ref{eq:OP_PA_NOMA1_P_j^*}).
	The second derivative of the objective function of $g(P_j)$ becomes
	\begin{equation}
	\frac{\mathrm d^2 g}{\mathrm d P_j^2} = \frac{N \Gamma_i^2 \Gamma_j^2 (\tau Q_j + \tilde{Z}_j - \tau Q_i - \tilde{Z}_i)^2}{(\Gamma_i - \Gamma_j)^2} \bigg( \frac{1}{\tau Q_j + \tilde{Z}_j} - \frac{1}{\tau Q_i + \tilde{Z}_i} \bigg).
	\end{equation}
	Since we already suppose that $1 < \frac{\tau Q_i + \tilde{Z}_i }{\tau Q_j + \tilde{Z}_j } < \frac{\Gamma_j}{\Gamma_i}$, so $P_j^* \geq 0$ and $\frac{\mathrm d^2 g}{\mathrm d P_j^2} > 0$, then $P_j^*$ is the global minimizer of $g(P_j)$.
	
	Consider the case $P_{j,(i,j)}^o \leq P_{i,(i,j)}^o$ first.
	The optimal point should be carefully found depending on the relative positions of (\ref{eq:OP_PA_NOMA1_PowerIntrval}), $\mathcal{O}_i$, and $\mathcal{O}_j$.
	The objective functions in $\mathcal{O}_j$ and $\mathcal{O}_i$ are $g_{\mathcal{O}_j}(P_j)$ and $g_{\mathcal{O}_i}(P_j)$, respectively.
	Therefore, three objective functions of $g(P_j)$, $g_{\mathcal{O}_j}(P_j)$ and $g_{\mathcal{O}_i}(P_j)$ are compared to determine the optimal power level, depending on the value of $P_j^*$.
	For example, consider the case of $q-P_0 \leq P_{j,(i,j)}^o$ and $P_{i,(i,j)}^o \leq \frac{q}{2}$, then $\mathcal{O}_j=[q-P_0, P_{j,(i,j)}^o]$ and $\mathcal{O}_i = [P_{i,(i,j)}^o, \frac{q}{2}]$.
	In addition, $g_{\mathcal{O}_j}(\bar{q})$ and $g_{\mathcal{O}_i}(q/2)$ are minimizers in $\mathcal{O}_j$ and $\mathcal{O}_i$, respectively.
	However, the minimizer in $[P_{j,(i,j)}^o,P_{i,(i,j)}^o]$ depends on $P_j^*$. 
	If $P_{j,(i,j)}^o \leq P_j^* \leq P_{i,(i,j)}^o$, $P_j^*$ minimizes $g(P_j)$ obviously.
	On the other hand, if $P_j^* < P_{j,(i,j)}^o$, $P_{j,(i,j)}^o$ becomes the minimizer of $g(P_j)$ in $[P_{j,(i,j)}^o,P_{i,(i,j)}^o]$. 
	Similarly, if $P_{i,(i,j)}^o < P_j^*$, $P_{i,(i,j)}^o$ is the minimizer of $g(P_j)$ in $[P_{j,(i,j)}^o,P_{i,(i,j)}^o]$. 
	Therefore, the optimal objective value in $[P_{j,(i,j)}^o, P_{i,(i,j)}^o]$ is given by $g(x)$, where
	\begin{equation}
	x = 
	\begin{cases}
		P_{j,(i,j)}^o, & \text{if } P_j^* < P_{j,(i,j)}^o \\
		P_j^*, & \text{if } P_{j,(i,j)}^o \leq P_j^* \leq P_{i,(i,j)}^o \\
		P_{i,(i,j)}^o, & \text{if } P_{i,(i,j)}^o < P_j^*
	\end{cases}.
	\end{equation}
	Finally, the globally optimal objective can be obtained by taking the minimal one among $g(x)$, $g_{\mathcal{O}_j}(\bar{q})$ and $g_{\mathcal{O}_i}(q/2)$, as shown in (\ref{eq:thm2_sol1}). 
	In a similar way, the solutions to (\ref{eq:thm2_sol2})-(\ref{eq:thm2_sol5}) can be also derived for the other cases depending on the relative positions of (\ref{eq:OP_PA_NOMA1_PowerIntrval}), $\mathcal{O}_i$, and $\mathcal{O}_j$.
	
	Next, at least one of users experiences the link outage in the case of $P_{i,(i,j)}^o \leq P_{j,(i,j)}^o$.
	We can define intervals of $\mathcal{O}_i = [P_{j,(i,j)}^o,P_0]$, $\mathcal{O}_j = [0,P_{i,(i,j)}^o]$, and $\mathcal{O}_{(i,j)} = [P_{i,(i,j)}^o, P_{j,(i,j)}^o]$ as the outage regions of user $i$, user $j$, and both users, respectively.
	If $P_{j,(i,j)}^N \in \mathcal{O}_{(i,j)}$, both links are in outage so NOMA becomes meaningless.
	Therefore, we just need to compare the objective function values of $g_{\mathcal{O}_i}$ and $g_{\mathcal{O}_j}$.
	For example, when $\bar{q} \leq P_{i,(i,j)}^o$, the minimal objective function value in $\mathcal{O}_j$ is $g_{\mathcal{O}_j}(\bar{q})$.
	On the other hand, when $P_{j,(i,j)}^o \leq \frac{q}{2}$, the minimal objective function value in $\mathcal{O}_i$ is $g_{\mathcal{O}_i}(\frac{q}{2})$.
	Thus, the solution to (\ref{eq:thm2_sol7}) is obtained by choosing the minimum of $g_{\mathcal{O}_j}(\bar{q})$ and $g_{\mathcal{O}_i}(\frac{q}{2})$.
	In addition, $P_{i,(i,j)}^o < \bar{q}$ and $\frac{q}{2} < P_{j,(i,j)}^o$ mean $P_{j,(i,j)}^N$ cannot be in $\mathcal{O}_j$ and $\mathcal{O}_i$, respectively, so (\ref{eq:thm2_sol8}) and (\ref{eq:thm2_sol9}) can be directly obtained.
	
\end{proof}

\textit{Remark}: Even though two users are paired for NOMA, no transmit power could be allocated to one of users.
This case indicates that the resource of one of users is taken by the other one. 
For example, when the outage occurs at a certain link, the resource of this link is preferred to be exploited by another link for resource efficiency.
Thus, we can see that finding the optimal power in this model enables user scheduling as well as flexible use of system resources.

The second subproblem for resolving the power allocation problem of (\ref{eq:OP_PA_NOMA})-(\ref{eq:OP_PA_NOMA_const}) finds the optimal auxiliary variable of $q$, to minimize the optimization metric of (\ref{eq:OP_PA_NOMA}), as follows:
\begin{align}
q^N = ~&\underset{q}{\arg \min} ~h(q) \label{eq:OP_PA_NOMA2}\\
\text{s.t.}~~ &0 \leq q \leq 2P_0, \label{eq:OP_PA_NOMA2_const}
\end{align}
where 
\begin{align}
h(q) &= V\cdot q - \frac{2\Phi\mathcal{B}(\tau Q_i + \tilde{Z}_i)}{N} \log_2 \Big( \frac{N \Gamma_i q /2 + 1}{N \Gamma_i P_j^N / 2 + 1} \Big) - \frac{2\Phi\mathcal{B}(\tau Q_j + \tilde{Z}_j)}{N} \log_2 \Big( 1 + N \Gamma_j P_j^N/2 \Big).
\end{align}

Differentiating $h(q)$ by $q$, 
\begin{equation}
\frac{\mathrm d h}{\mathrm d q} = V - \frac{\tau Q_i + Z_i}{\ln 2} \cdot \frac{\Phi\mathcal{B}\Gamma_i}{N\Gamma_i q/2 +1},
\end{equation}
and the local minimizer $q^*$ can be obtained from $\frac{\mathrm d h}{\mathrm d q} = 0$, as written by
\begin{equation}
q^* = \frac{2}{N\Gamma_i} \Big( \frac{\Phi\mathcal{B}\Gamma_i}{V \ln 2}(\tau Q_i + Z_i) - 1 \Big).
\end{equation}
Since 
\begin{equation}
\frac{\mathrm d^2 h}{\mathrm d q^2} = \frac{\tau Q_i + Z_i}{\ln 2} \cdot \frac{\mathcal{B}\Gamma_i}{(N\Gamma_i q/2 + 1)^2} \cdot \frac{N\Gamma_i}{2} > 0,
\end{equation}
$q^*$ is the global minimizer of $h(q)$.
Considering the constraint (\ref{eq:OP_PA_NOMA2_const}), the optimal $q$ is obtained by
\begin{equation}
q^N = 
\begin{cases}
0, &~ \text{if}~ q^* < 0 \\
q^*, &~ \text{if}~ 0 \leq q^* \leq 2P_0 \\
2P_0, &~ \text{if}~ 2P_0 < q^*
\end{cases}.
\label{eq:sol_PowerSubprob2}
\end{equation}
Herein, $q^N=0$ makes NOMA useless, because both users will be in outage.

Thus, the non-convex optimization problem of (\ref{eq:OP_PA_NOMA})-(\ref{eq:OP_PA_NOMA_const}) for finding the optimal power allocation for NOMA users can be solved by dealing with two convex subproblems of (\ref{eq:OP_PA_NOMA2})-(\ref{eq:OP_PA_NOMA2_const}) and (\ref{eq:OP_PA_NOMA1})-(\ref{eq:OP_PA_NOMA1_const2}) in sequential. 
The transmitter can first optimize the transmit power consumption for a user pair, $q^N$, based on the current CSI and QSI.
Then, power levels allocated to NOMA users, $P_{i,(i,j)}^N$ and $P_{j,(i,j)}^N$, can be achieved by Theorem \ref{thm2}.

\section{Matching Algorithm for NOMA User Pairing}
\label{sec:user_pairing}

Since the optimal power allocation for NOMA users is derived when a pair of NOMA users is already determined, there remains the problem of which users are better to be paired for NOMA signaling.
User pairing can be interpreted as a kind of matching. 
We define the matching $\Psi$ which indicates user pairings for NOMA signaling.
\begin{definition}
	A matching $\Psi$ is defined by (\ref{eq:matching_const1})-(\ref{eq:matching_const3}) as follows:
	\begin{align}
	&\Psi (u_i) \in \mathcal{U},~\forall u_i \in \mathcal{U} \label{eq:matching_const1} \\
	&|\Psi (u_i)| = 1 \label{eq:matching_const2} \\
	&\Psi (u_i) = u_j \iff \Psi (u_j) = u_i \label{eq:matching_const3} 
	\end{align}
\end{definition}

Specifically, $\Psi(u_i)$ indicates the user paired with user $u_i$ and both users are in the same user set $\mathcal{U}$ consisting of $N$ users, so (\ref{eq:matching_const1}) is satisfied.
$\Psi(u_i)=u_i$ means that OMA is used for $u_i$, and $\Psi(u_i) = u_j$ for $i \neq j$ indicates that $u_i$ and $u_j$ are paired for NOMA.
Since we considered the two-user model, (\ref{eq:matching_const2}) is provided and $\Psi$ becomes the one-to-one matching.
When $u_i$ and $u_j$ are paired, both $\Psi(u_i) = u_j$ and $\Psi(u_j) = u_i$ are satisfied as shown in (\ref{eq:matching_const3}), so $\Psi = \Psi^{-1}$.

The matching $\Psi$ is constructed according to the preference lists of users. 
Denote the preference list of $u_i$ by $\mathcal{P}_i$ for all $u_i \in \mathcal{U}$.
When $\mathcal{M}_{i,(i,j)}^N < \mathcal{M}_{i,(i,k)}^N$, $u_i$ prefers to be paired with $u_j$ to $u_k$.
Herein, $\mathcal{M}_{i,(i,i)}^N = \mathcal{M}_i^O$.
In addition, we only allow $u_j$ to be included in $\mathcal{P}_i$ when $\mathcal{M}_{i,(i,j)}^N \leq 0$.
The reason is that $\mathcal{M}_{i,(i,j)}^N = 0$ for any $u_j$ is obtained when $P_i = 0$. 
Before constructing $\Psi$, each optimization metric $\mathcal{M}_{i,(i,j)}^N$, for all $i,j \in \mathcal{U}$ can be obtained by solving the problems of (\ref{eq:OP_PA_OMA})-(\ref{eq:OP_PA_OMA_const}) and (\ref{eq:OP_PA_NOMA})-(\ref{eq:OP_PA_NOMA_const}) for $i = j$ and $i \neq j$, respectively.
For given $\Psi$, the total optimization metric of (\ref{eq:Lyapunov_metric}) can be computed as $\mathcal{M}(\mathbf{P},\boldsymbol{\Psi}) = \sum_{i \in \mathcal{N}} \mathcal{M}_{i,(i,\Psi(u_i))}^N$.

Since it takes too much complexity to compute and compare optimization metrics for all possible pairing combinations, we focus on seeking the stable matching by using the deferred acceptance (DA) procedure \cite{Book:Roth}.
Each user sends the matching request to the most preferred user, and the user who receives the request can accept or reject the pairing with the sender. 
The user pairing algorithm is shown in Algorithm \ref{algo:user_pairing}, and the details of matching request and decision for the received request are expressed in Algorithm \ref{algo:matching_request}.

For example, suppose that $u_i$ sends the matching request to $u_j$ in the matching $\Psi$.
Let $\Psi'$ be the \textit{optimal} matching when $u_j$ accepts the request from $u_i$.
Then, $u_j$ decides whether to accept or reject the request from $u_i$ by comparing 
$\mathcal{M}(\mathbf{P}(\boldsymbol{\Psi}),\boldsymbol{\Psi})$ and $\mathcal{M}(\mathbf{P}(\boldsymbol{\Psi'}),\boldsymbol{\Psi'})$.
$\mathcal{M}(\mathbf{P}(\boldsymbol{\Psi}),\boldsymbol{\Psi})$ is already obtained with the current matching $\Psi$, and we need to compute $\mathcal{M}(\mathbf{P}(\boldsymbol{\Psi'}),\boldsymbol{\Psi'})$.
When $\Psi(u_i) = u_i$ and $\Psi(u_j)=u_j$, the matching request is simply accepted when $\mathcal{M}_{i,(i,j)}^N (P_{i,(i,j)}^N) + \mathcal{M}_{j,(i,j)}^N (P_{j,(i,j)}^N) < \mathcal{M}_i^O (P_i^O) + \mathcal{M}_j^O (P_j^O)$. Then the \textit{optimal} matching becomes 
$\Psi' = \Psi \setminus \{(u_i,u_i),(u_j,u_j)\} \cup \{(u_i,u_j),(u_j,u_i)\}$.

However, when $\Psi(u_i) \neq u_i$ and/or $\Psi(u_j) \neq u_j$, if $u_j$ accepts the matching request from $u_i$, $\Psi(u_i)$ and/or $\Psi(u_j)$ should find another pair to construct the \textit{optimal} matching $\Psi'$.
According to Algorithm \ref{algo:matching_request}, $\Psi(u_i)$ and $\Psi(u_j)$ send the matching request to their most preferred users, except for $u_i$ and $u_j$.
If the most preferred users of $\Psi(u_i)$ and $\Psi(u_j)$ are themselves respectively, 
then $\Psi' = \Psi \setminus \{ (u_i,\Psi(u_i)), (u_j,\Psi(u_j)) \} \cup \{ (u_i,u_j), (u_j,u_i), (\Psi(u_i), \Psi(u_i)), (\Psi(u_j),\Psi(u_j)) \}$.
If not, Algorithm \ref{algo:matching_request} should be recursively performed to construct $\Psi'$, until all users are matched.
Finally, compute $\mathcal{M}(\mathbf{P}(\boldsymbol{\Psi'}),\boldsymbol{\Psi'})$ and compare it with $\mathcal{M}(\mathbf{P}(\boldsymbol{\Psi}),\boldsymbol{\Psi})$.
If $\mathcal{M}(\mathbf{P}(\boldsymbol{\Psi}),\boldsymbol{\Psi}) > \mathcal{M}(\mathbf{P}(\boldsymbol{\Psi'}),\boldsymbol{\Psi'})$, the match request from $u_i$ to $u_j$ is accepted and $\Psi$ is updated by $\Psi'$, as shown in Algorithm \ref{algo:user_pairing}.

\begin{algorithm}
	\caption{User pairing algorithm for NOMA transmissions
		\label{algo:user_pairing}}
	\begin{algorithmic}[1]
		\State{Initialize $\Psi(u) \leftarrow u,~\forall u \in \mathcal{U}$.}
		\For{ $\forall u_i \in \mathcal{U}$}
		{
			\State{$\mathcal{F} \leftarrow \phi$}
			\While {true}
				\State{Find $u_j \leftarrow \underset{u \in \mathcal{P}_i\setminus \mathcal{F}}{\arg \min}~\mathcal{M}_{i,(i,j)}^N$}
				\If{$j==\Psi(u_i)$}\State{break;}\EndIf
				\State{$\mathcal{F} \leftarrow \mathcal{F} \cup \{u_j\}$}
				\State{$\Psi' \leftarrow \textit{MatchRequest}(u_i,u_j,\Psi,\phi)$}
				\If{$\mathcal{M}(\mathbf{P}(\boldsymbol{\Psi}),\boldsymbol{\Psi}) > \mathcal{M}(\mathbf{P}(\boldsymbol{\Psi'}),\boldsymbol{\Psi'})$}
					\State{$\Psi \leftarrow \Psi'$}
					\State{break;}
				\EndIf
			\EndWhile
		}
		\EndFor
	\end{algorithmic}
\end{algorithm}

\begin{algorithm}
	\caption{Matching request algorithm
		\label{algo:matching_request}}
	\begin{algorithmic}[1]
		\State{$\mathbf{Input}$: $u_i$, $u_j$, $\Psi'$, and $\mathcal{E}$.}
		\State{$\mathbf{Output}$: $\Psi'$}
		\State{$m \leftarrow \Psi'(u_i)$ and $p \leftarrow \Psi'(u_j)$}
		\State{$\Psi(u_i) \leftarrow u_j$ and $\Psi(u_j) \leftarrow u_i$}
		\If{$i \neq m$} $\Psi'(u_m) \leftarrow u_m$ \EndIf
		\If{$j \neq p$} $\Psi'(u_p) \leftarrow u_p$ \EndIf
		\State{$\mathcal{E}_m \leftarrow \mathcal{E} \cup \{u_i,u_j\}$}
		\If{$i \neq m$}
		\State{Find $n \leftarrow \underset{u_n \in \mathcal{P}_m\setminus \mathcal{E}_m}{\arg \min}~\mathcal{M}_{m,(m,n)}^N$}
		\If{$n == m$} $\Psi'(u_n) \leftarrow u_n$
		\ElsIf{$n == p$} $\Psi'(u_n) \leftarrow u_p$ and $\Psi'(u_p) \leftarrow u_n$ 
		\Else $~\Psi' \leftarrow \textit{MatchRequest}(u_m,u_n,\Psi',\mathcal{E}_m)$
		\EndIf
		\EndIf
		
		\State{$\mathcal{E}_p \leftarrow \mathcal{E} \cup \{u_i,u_j\}$}
		\If{$j \neq p~\&\&~\Psi'(u_p) == u_p$}
		\State{Find $q \leftarrow \underset{u_q \in \mathcal{P}_p\setminus \mathcal{E}_p}{\arg \min}~\mathcal{M}_{p,(p,q)}^N$}
		\If{$q == p$} $\Psi'(u_q) \leftarrow u_q$
		\Else $~\Psi' \leftarrow \textit{MatchRequest}(u_p,u_q,\Psi',\mathcal{E}_p)$
		\EndIf
		\EndIf
		\State{$\mathbf{return}~\Psi'$}
	\end{algorithmic}
\end{algorithm}

The optimal user pairing can be obtained by searching over all possible combinations of user pairing to find $\boldsymbol{\Psi}$ that maximizes $\mathcal{M}(\mathbf{P}(\boldsymbol{\Psi}),\boldsymbol{\Psi})$.
Suppose that $L$ pairs are allowed for NOMA.
Then, the transmitter needs to exhaustively search $\prod_{l=1}^{L} \binom{N-2(l-1)}{2}$ combinations for the optimal user pairing, and time complexity is approximately $O(N^{2L})$.
In the proposed algorithm, the worst case is that no pair is generated for $u_1,\cdots,u_{N-L}$ and then 
a new pair is matched every time for the last $L$ users. 
It requires $\sum_{l=1}^{L}N-2(l-1)+N(N-L)$ comparison steps and time complexity is $O(N^2)$.
Thus, the complexity of the proposed matching algorithm is much less than that of the optimal user pairing.
Note that the complexity gain of the proposed algorithm grows with large $N$ and $L$.

\section{Performance Evaluation}
\label{sec:simulation}

Our data-intensive simulations for performance evaluation are based on the cellular model with a radius $R=50$. 
There exist $N=40$ users and all users are uniformly placed in the whole cellular region.
Assume that $\mathcal{N}_s = \mathcal{N}$, $\rho = \rho_n$ and $\eta = \eta_n$ for all $n \in \mathcal{N}$.
All parameters are listed in Table \ref{table:parameters}, and these are used unless otherwise noted.
A short frame structure designed for URLLCs in 5G networks \cite{ShortPacket:VTC2015Kela} is used for simulation results.
Note that the maximum queueing delay bound is $D_{max}^q = 0.9$ ms, and we will show that the proposed algorithm achieves this delay constraint.

To verify the advantages of the proposed algorithm, this paper compares the proposed one with other schemes:
\begin{itemize}
	\item `pMax': The transmitter always consumes the maximum power budget for all of $N$ users, except when the link outage occurs.
	NOMA and user pairing are not considered.
	
	\item `pMin': The transmitter always consumes the minimum power to avoid the link outage.
	If the required power for avoiding the outage is greater than the power budget, the link remains in outage.
	NOMA and user pairing are not considered.
	
	\item `opt. OMA': The power allocation is based on the proposed optimization framework (\ref{eq:OP1})-(\ref{eq:OP1_Powerconst}) but NOMA and user pairing are not considered.
\end{itemize}
To emphasize the difference from comparison schemes, we will call the proposed scheme `opt. hybrid MA'.

\begin{table}[t!]%
	\caption{System parameters \cite{ShortPacket:VTC2015Kela}, \cite{QueueDelay:VTC2016She}}
	\label{table:parameters}
	{\footnotesize
		\begin{center}
			\begin{tabular}{|p{5.5cm}|c|p{5.0cm}|c|}
				\hline
				E2E delay bound ($D_{max}$) & 1 ms & Power budget for each user ($P_0$) & 3 W \\
				\hline
				Frame duration ($T_f$) & 0.1 ms\ & User number ($N$) & 40 \\
				\hline
				DL phase duration ($T_D$) & 0.05 ms & Bandwidth ($BW$) & 20 MHz \\
				\hline
				Maximal queueing delay ($D_{max}^q$) & 0.9 ms & $V$ & $5\times 10^5$ \\ 
				\hline
				Packet size ($u$) & 160 bits & $\lambda_{min}$ & 5 packets \\ 
				\hline 
				Cell radius ($R$) & 50 m & $\lambda_{max}$ & 10 packets \\ 
				\hline
				Path loss model & 35.3+37.6($d_k$) & $\rho$ & 7 Mb \\
				\hline
				Single-sided noise spectral density ($N_0$) & -173 dBm/Hz & $\eta$ & 8.5 Mb \\
				\hline
			\end{tabular}
		\end{center}
	}
\end{table}

\begin{figure}[t]
	\minipage{0.45\textwidth}
	\includegraphics[width=\linewidth]{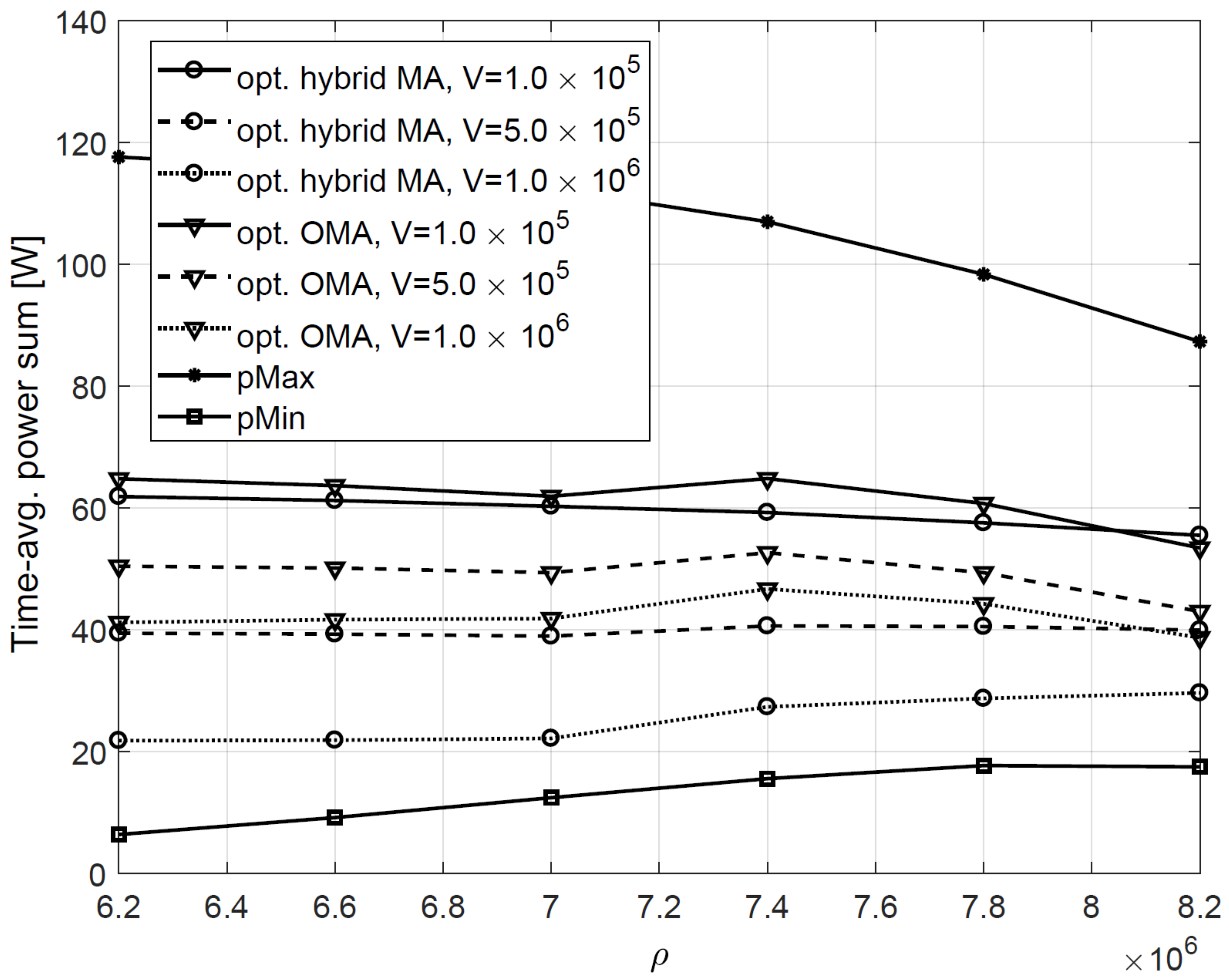}
	\caption{Time-avg. transmit power sum vs. $\rho$} \label{fig:power_center_rholist}
	\endminipage\hfill
	\minipage{0.45\textwidth}
	\includegraphics[width=\linewidth]{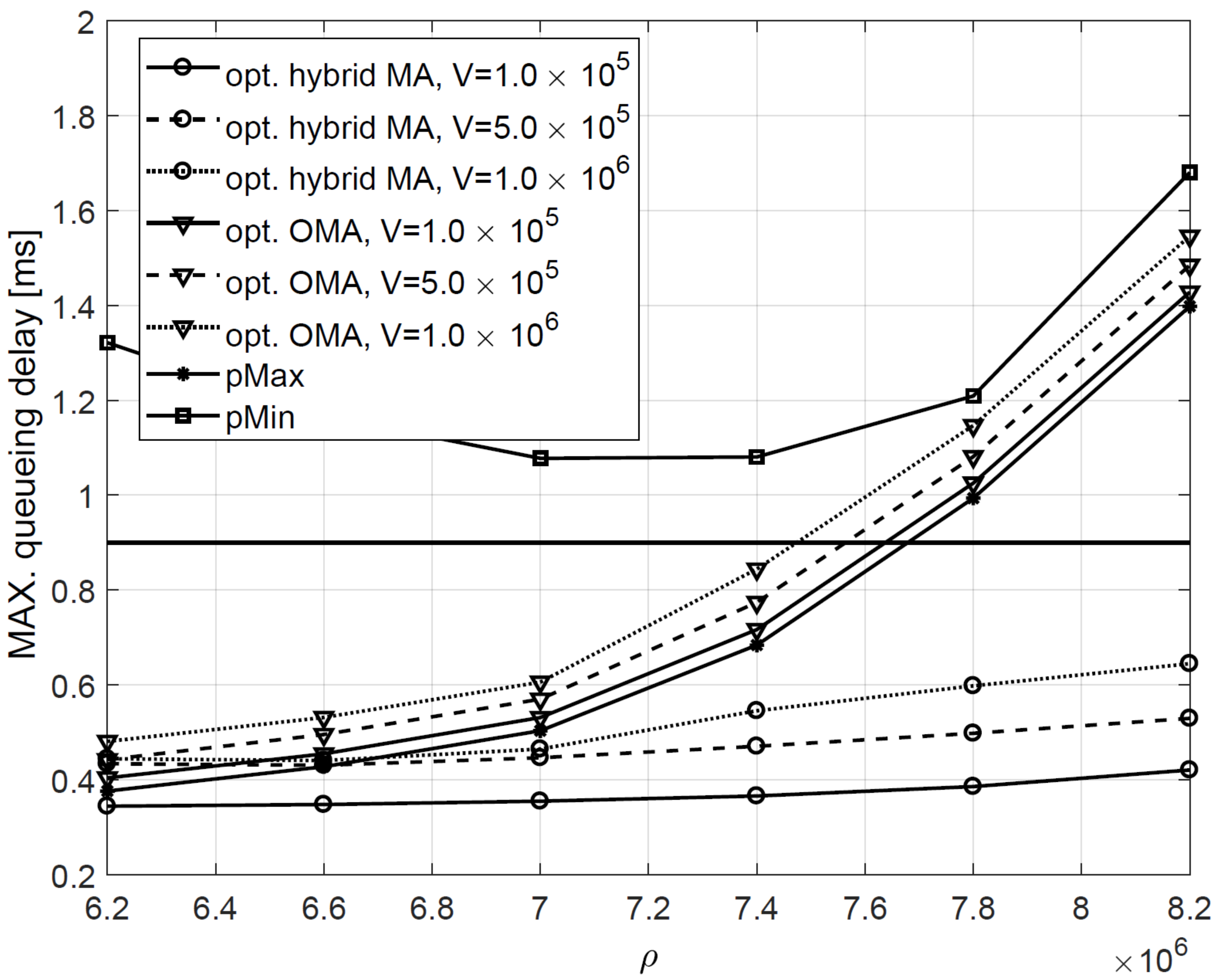}
	\caption{Maximal queueing delay vs. $\rho$} \label{fig:delay_center_rholist}
	\endminipage
\end{figure}

\begin{figure}[t]
	\minipage{0.45\textwidth}
	\includegraphics[width=\linewidth]{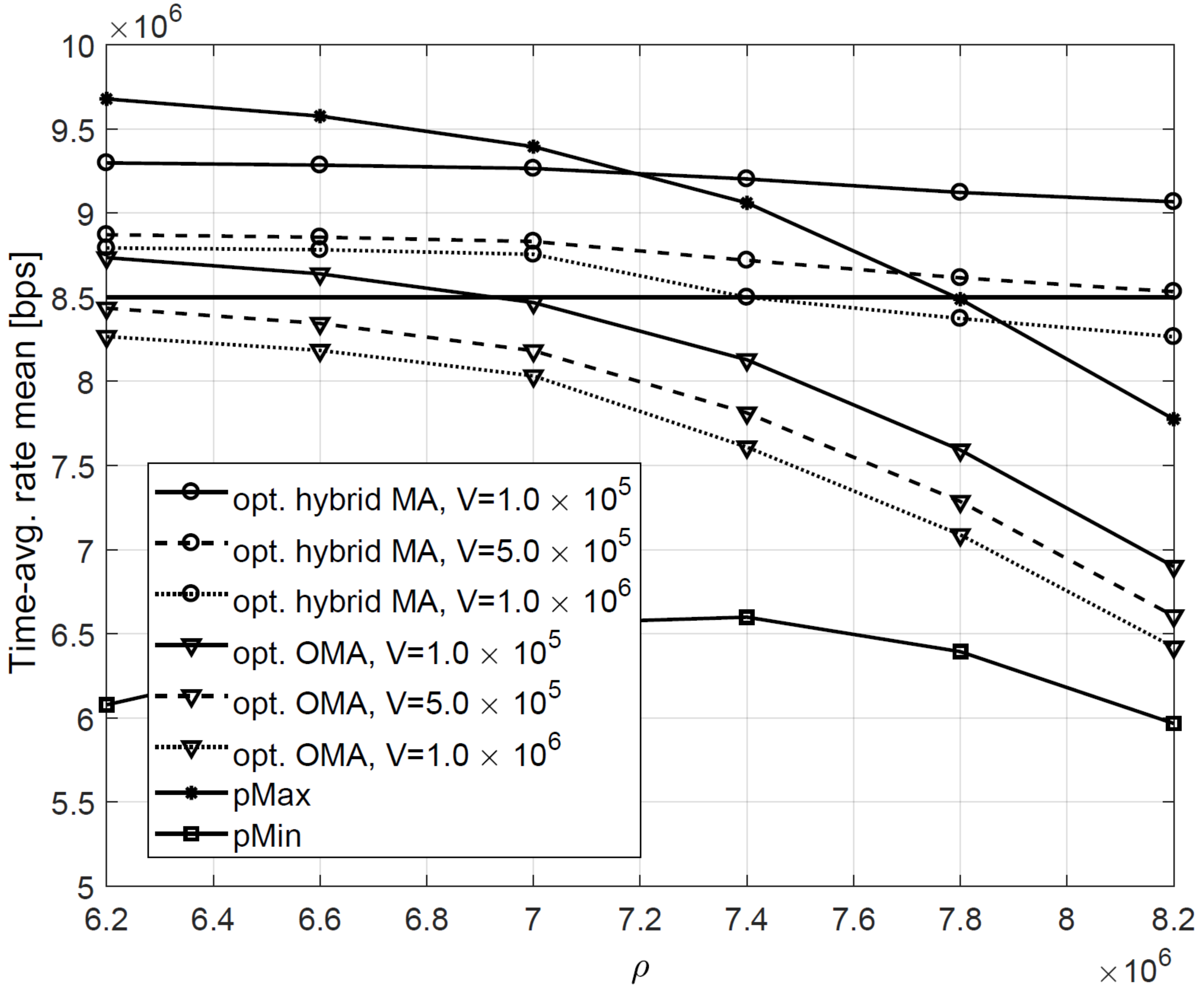}
	\caption{Time-avg. data rate vs. $\rho$} \label{fig:datarate_center_rholist}
	\endminipage\hfill
	\minipage{0.45\textwidth}
	\includegraphics[width=\linewidth]{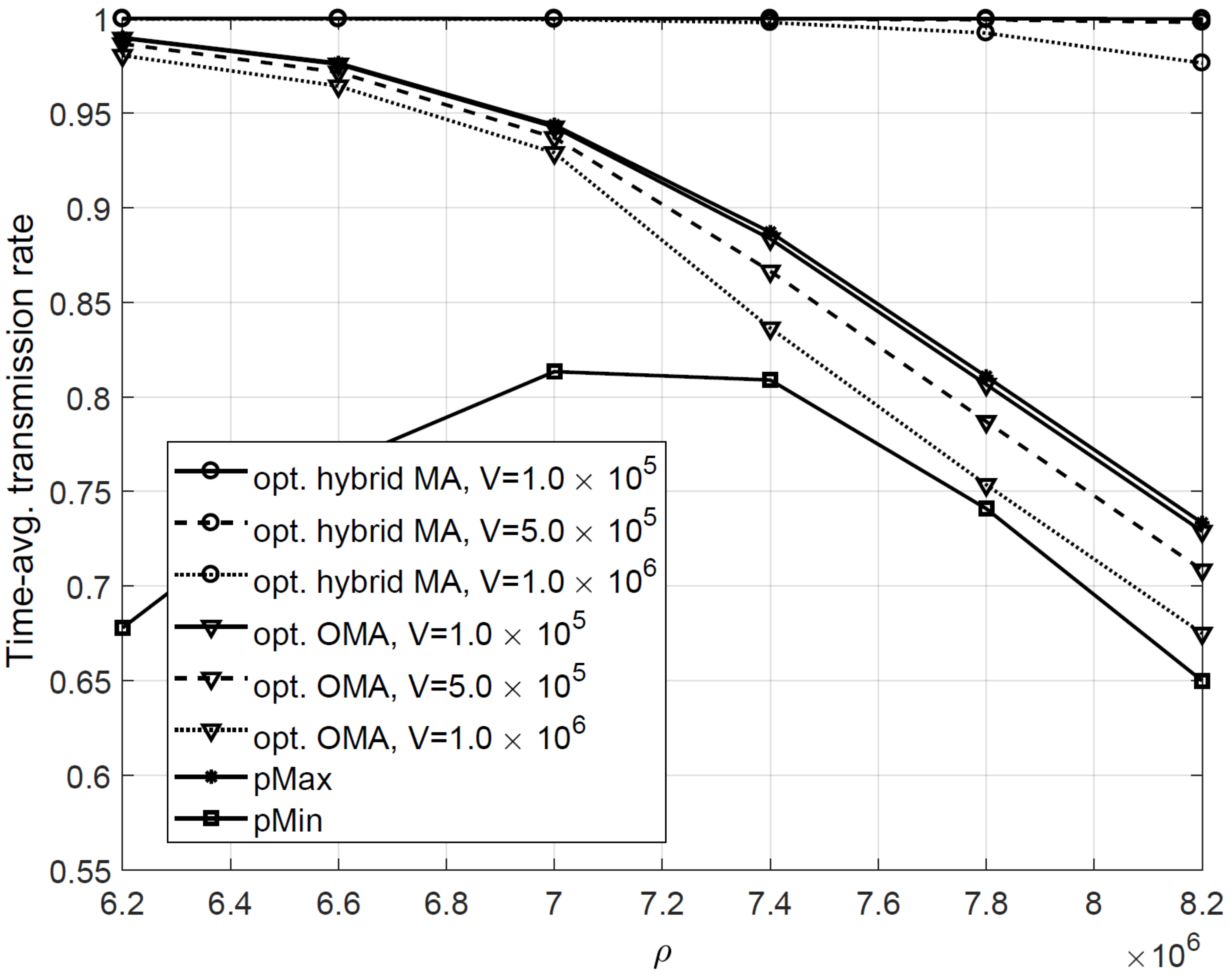}
	\caption{Low latency transmission rate vs. $\rho$} \label{fig:transrate_center_rholist}
	\endminipage
\end{figure}

Figs. \ref{fig:power_center_rholist} and \ref{fig:delay_center_rholist} show plots of time-average transmit power consumption for $N$ users and the expected maximum queueing delay among $N$ users versus $\rho$, i.e., outage threshold, respectively. 
In addition, the time-average data rates and low-latency transmission rates with different values of $\rho$ are shown in Figs. \ref{fig:datarate_center_rholist} and \ref{fig:transrate_center_rholist}, respectively. 
The low-latency transmission rate is the probability that a data packet arriving at the queue can be transmitted within a time slot of $D_{max}^q$.
The link outage occurs more frequently as $\rho$ increases.
Therefore, increasing $\rho$ basically makes power consumption and queueing delay grow, and time-average data rate and low-latency transmission rate decrease. 

However, some peculiar trends are observed in Figs. \ref{fig:power_center_rholist}-\ref{fig:transrate_center_rholist}.
First, some cases show decreasing power consumption as $\rho$ increases, and this results from frequent link deactivations due to link outage occurrences. 
Second, the performance trends of pMin are not monotonic in Figs. \ref{fig:delay_center_rholist}-\ref{fig:transrate_center_rholist}.
The reason is that the activated link of pMin always provides the data rate of $\rho$, because pMin consumes the minimum power to avoid the outage.
Therefore, the time-average data rate of pMin increases with small $\rho$, but decreases when $\rho$ is large due to frequent outage occurrences.

Among comparison techniques, we can notice that the proposed opt. hybrid MA provides high power efficiency while guaranteeing low queueing delay and sufficient time-average data rate.
opt. hybrid MA with $V=1.0 \times 10^6$ consumes almost the same power as pMin in Fig. \ref{fig:power_center_rholist}, as well as the queueing delay of opt. hybrid MA with $V=1.0 \times 10^5$ is the shortest among comparison schemes in Fig. \ref{fig:delay_center_rholist}.
Especially when $\rho$ is large, all the other schemes require maximum queueing delays larger than $D_{max}^q$, but opt. hybrid MA does not.
Therefore, opt. hybrid MA also shows the best low-latency transmission rates given in Fig. \ref{fig:transrate_center_rholist}.
In addition, opt. hybrid MA satisfies the QoS constraint as shown in Fig. \ref{fig:datarate_center_rholist}.

We can also find the advantages of NOMA over OMA by comparing opt. hybrid MA with opt. OMA. 
NOMA is well-known to improve throughput compared to OMA, with the same power consumption.
We can see that opt. hybrid MA gives better data rates with smaller power consumption than opt. OMA in Figs. \ref{fig:power_center_rholist} and \ref{fig:datarate_center_rholist}. 
Since the time-average data rate is sufficient to guarantee $\eta$, the system flexibly utilizes NOMA advantages over OMA in terms of both data rate and power efficiency.
Further, lower queueing delays and higher low-latency transmission rates are achieved by using NOMA with appropriate user pairings and power allocations.

The effects of system parameter $V$ can be also shown in Figs. \ref{fig:power_center_rholist}-\ref{fig:transrate_center_rholist}. 
As we explained earlier, $V$ is a weight factor for the term representing the transmit power in (\ref{eq:Lyapunov_metric}). 
As $V$ becomes larger, opt. hybrid MA and opt. OMA further pursue power efficiency rather than reducing backlogs of actual and virtual queues, $Q_n(t)$ and $Z_n(t)$ for all $n\in \mathcal{N}$, respectively.
Therefore, the time-average transmit power decreases with $V$. 
On the other hand, backlogs of actual and virtual queues grow as $V$ increases, 
so the queueing delay usually increases, the low latency transmission rate decreases and the time-average data rate decreases for opt. hybrid MA and opt. OMA.
Thus, the tradeoff between transmit power and queueing delay can be controlled by adjusting the system parameter $V$, depending on stochastic networks and QoS requirements.

\begin{figure}[t]
	\minipage{0.45\textwidth}
	\includegraphics[width=\linewidth]{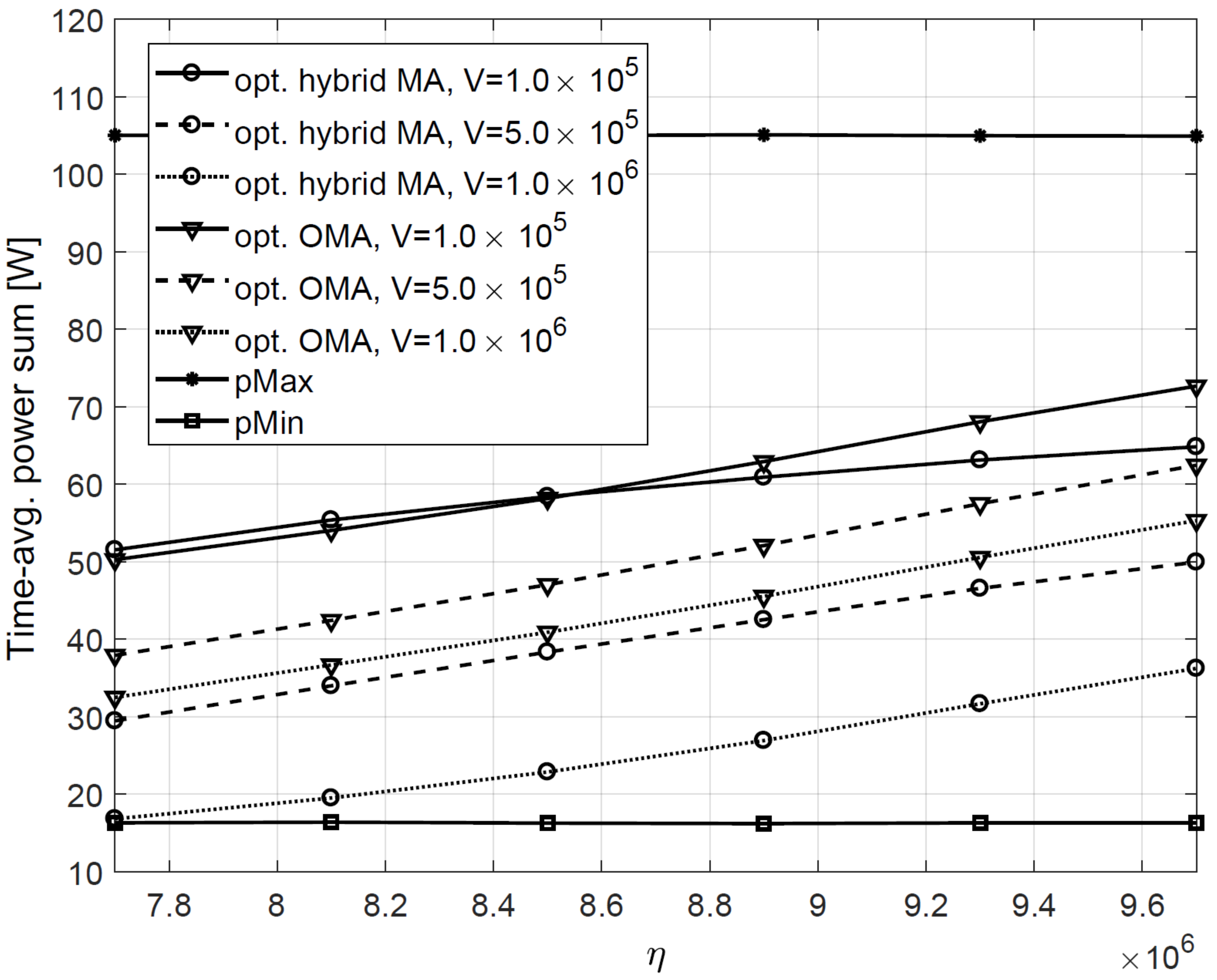}
	\caption{Time-avg. transmit power sum vs. $\eta$} \label{fig:power_center_nulist}
	\endminipage\hfill
	\minipage{0.45\textwidth}
	\includegraphics[width=\linewidth]{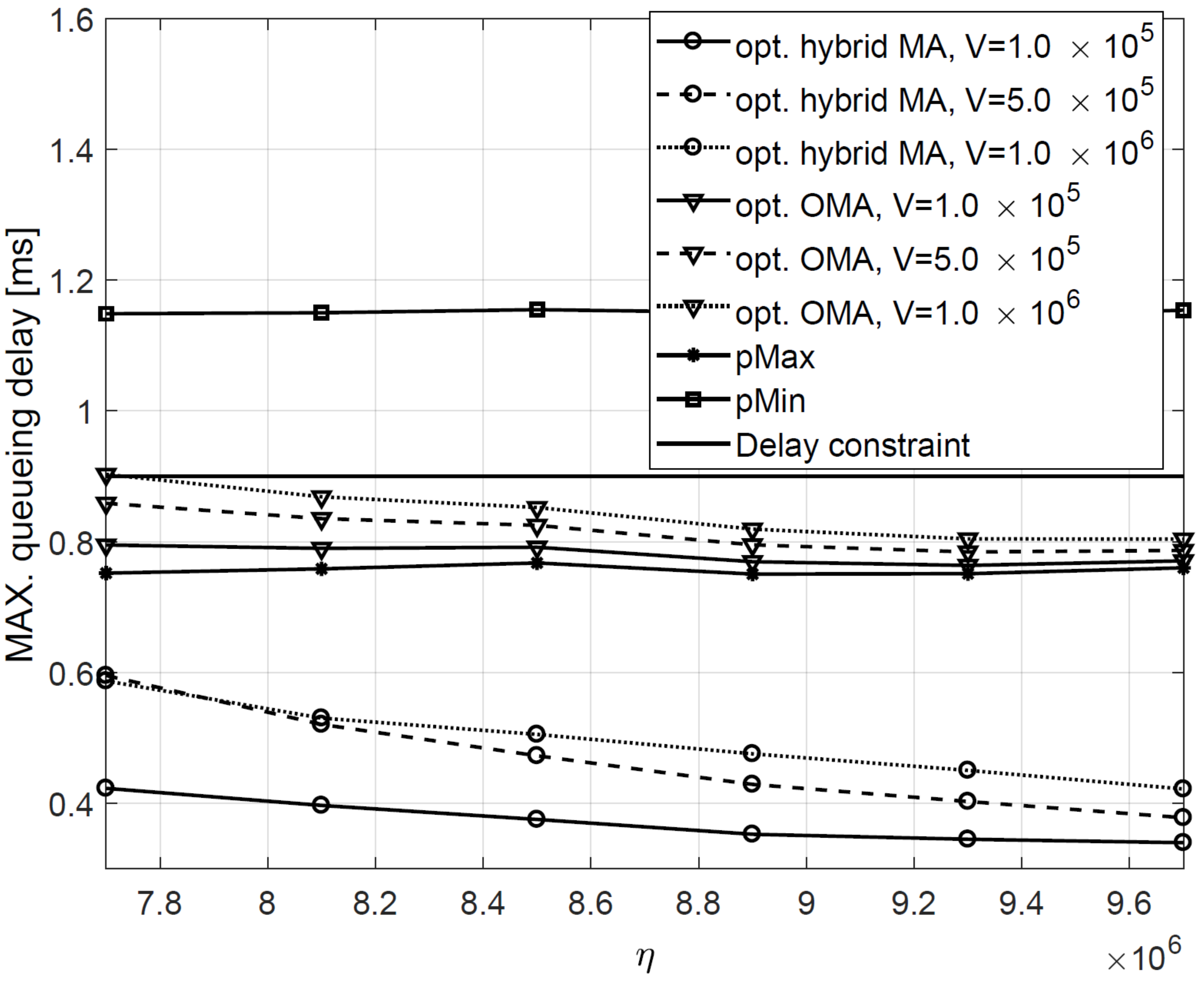}
	\caption{Maximal queueing delay vs. $\eta$} \label{fig:delay_center_nulist}
	\endminipage
\end{figure}

\begin{figure}[t]
	\minipage{0.45\textwidth}
	\includegraphics[width=\linewidth]{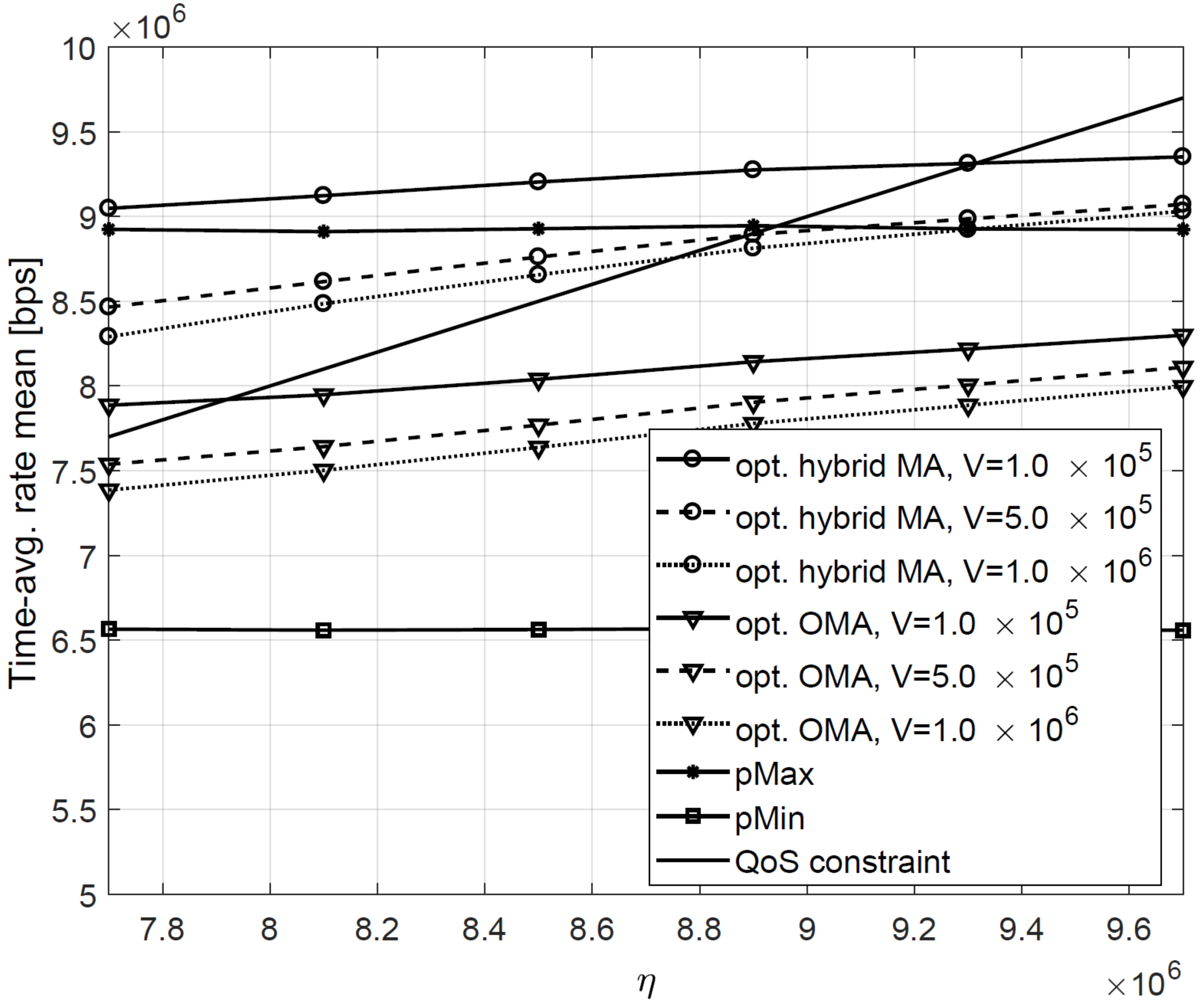}
	\caption{Time-avg. data rate vs. $\eta$} \label{fig:datarate_center_nulist}
	\endminipage\hfill
	\minipage{0.45\textwidth}
	\includegraphics[width=\linewidth]{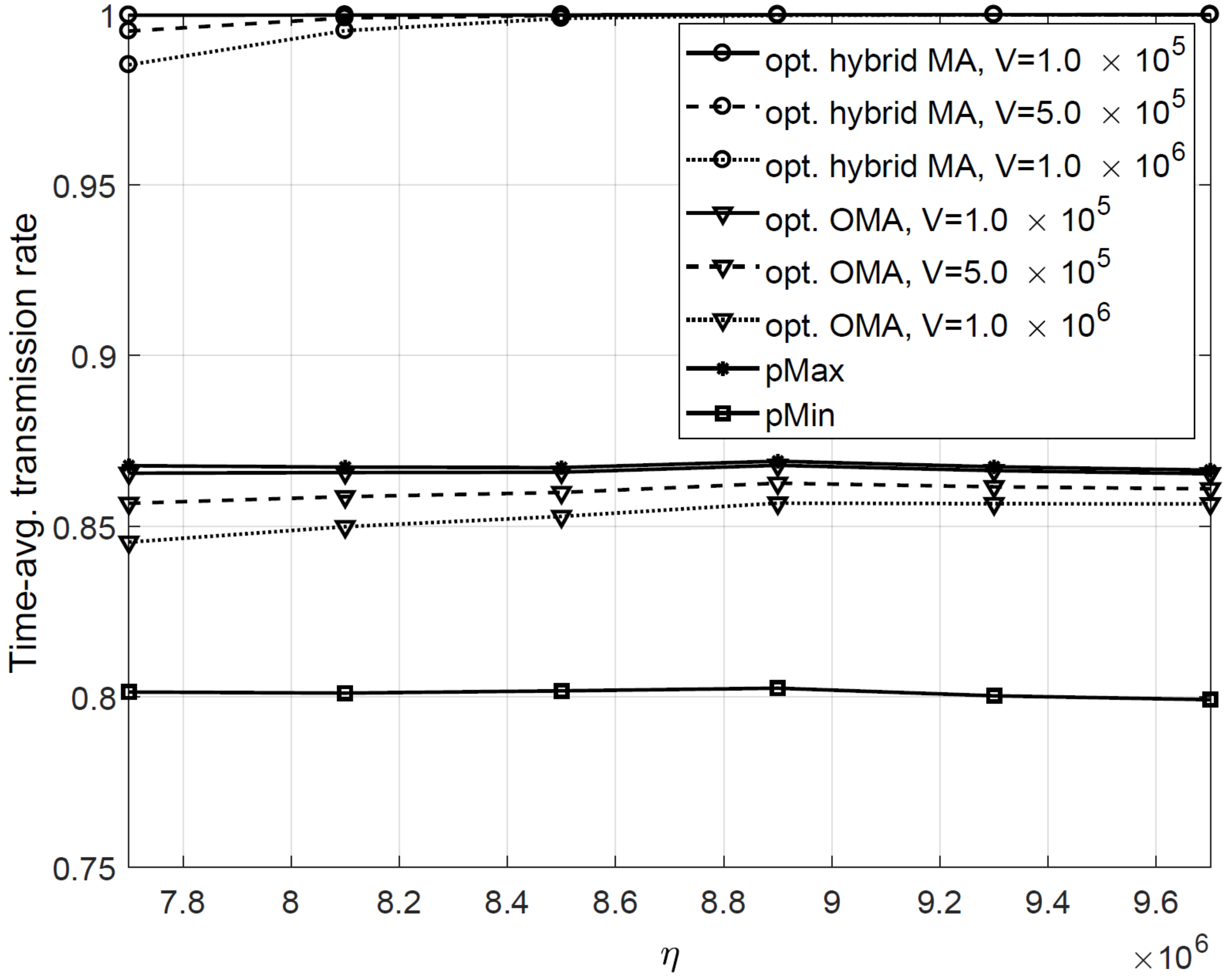}
	\caption{Low latency transmission rate vs. $\eta$} \label{fig:transrate_center_nulist}
	\endminipage
\end{figure}

Figs. \ref{fig:power_center_nulist}-\ref{fig:transrate_center_nulist} show plots of power consumption, queueing delay, time-average data rate and low latency transmission rate versus $\eta$, respectively.
The QoS constraint $\eta$ only affects performances of opt. hybrid MA and opt. OMA, because pMax and pMin do not consider the QoS constraint.
As $\eta$ grows, the transmit powers of opt. hybrid MA and opt. OMA obviously increase to guarantee $\eta$ as much as possible, so the time-average data rate increases for opt. hybrid MA and opt. OMA.
Since a larger power causes more departures, their queueing delays decrease also.
However, it becomes too difficult to satisfy QoS requirements when $\eta$ is large, as shown in Fig. \ref{fig:datarate_center_nulist}.
The encouraging point is that opt. hybrid MA can provide higher time-average data rates even with less power consumption than pMax.
In addition, comparing opt. hybrid MA with opt. OMA, we can see that NOMA advantages over OMA still remain for different values of $\eta$.
Similar to performance changes with $V$ in Figs. \ref{fig:power_center_rholist}-\ref{fig:transrate_center_rholist}, the time-average transmit power decreases as $V$ grows, whereas the queueing delay increases as shown in Figs. \ref{fig:power_center_nulist} and \ref{fig:delay_center_nulist}, respectively.

\section{Concluding Remarks}
\label{sec:conclusion}

This paper designs the joint optimization framework for power allocation and user pairing in hybrid MA system.
The optimization framework pursues both power efficiency and low latency while achieving sufficient time-average data rates for all users.
User pairings for NOMA signaling are performed based on the matching theory, and the closed-form optimal power allocations for OMA and NOMA users with a given policy of user pairing are derived.
The proposed algorithm dynamically performs user pairings for NOMA with the optimal power allocations to adjust backlogs in transmitter queues.
Based on the short frame structure which is designed for URLLCs in the 5G network, simulation results show that the proposed algorithm enables to achieve a E2E delay smaller than 1 ms, while guaranteeing high power-efficiency and sufficient time-average data rates.
The proposed dynamic power control and user pairing algorithm smooths out the tradeoff between power and delay, and can also control the tradeoff by adjusting system parameters.



\begin{IEEEbiography}[{\includegraphics[width=1in,height=1.25in,clip,keepaspectratio]{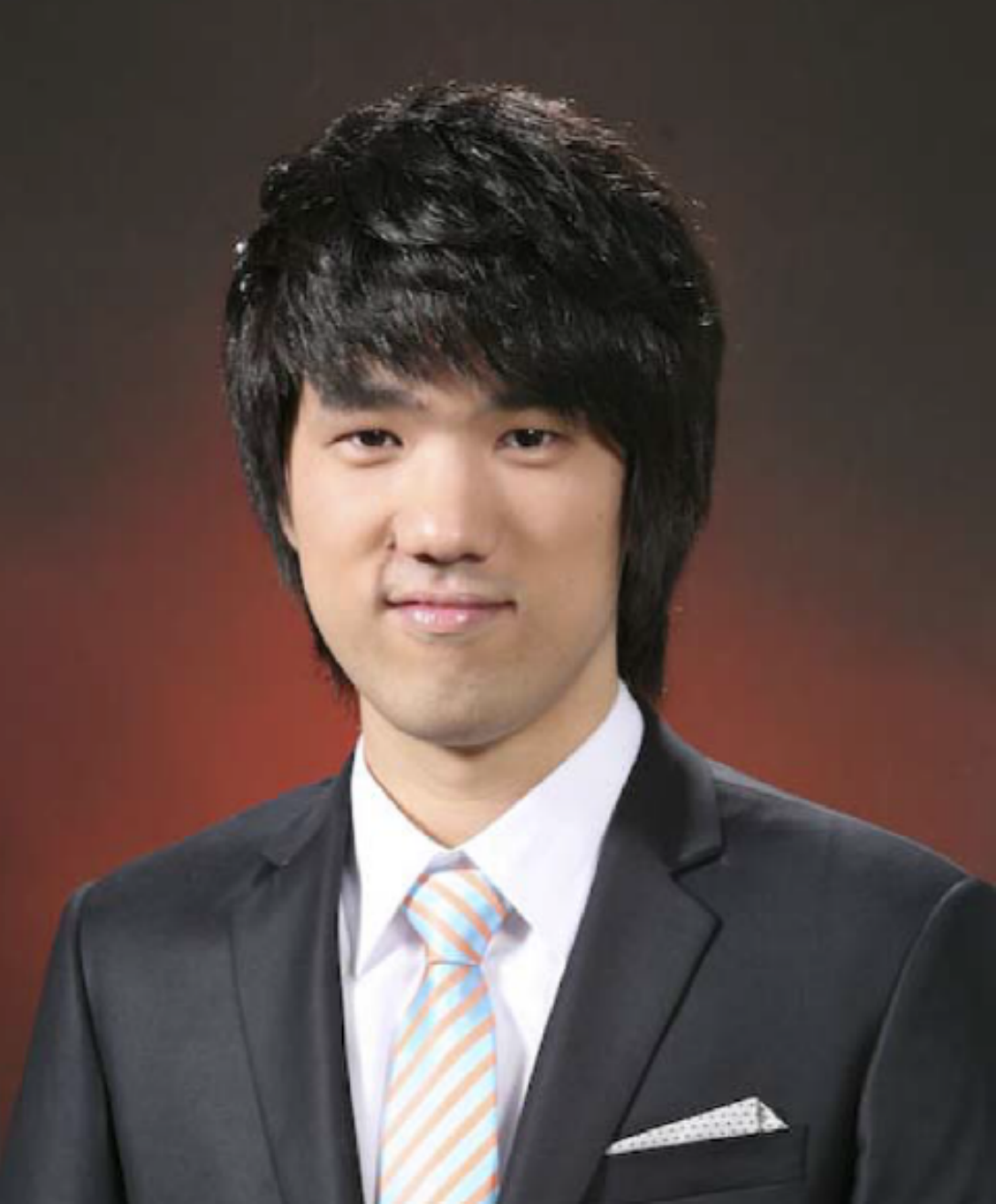}}]{Minseok Choi}
	received the B.S. and M.S. degree in electrical engineering from the Korea Advanced Institute of Science and Technology (KAIST), Daejeon, Korea, in 2016. He is currently pursuing the Ph.D degree in KAIST. His research interests include NOMA, wireless caching network, stochastic network optimization, and 5G network.\end{IEEEbiography}
\begin{IEEEbiography}[{\includegraphics[width=1in,height=1.25in,clip,keepaspectratio]{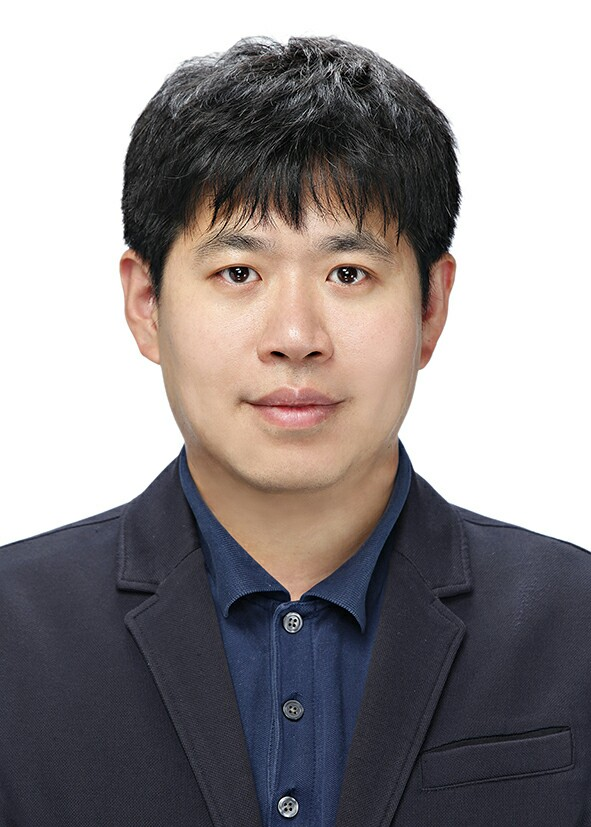}}]{Joongheon Kim}
	(M'06--SM'18) has been an assistant professor with Chung-Ang University, Seoul, Korea, since 2016. He received his B.S. (2004) and M.S. (2006) in computer science and engineering from Korea University, Seoul, Korea; and his Ph.D. (2014) in computer science from the University of Southern California (USC), Los Angeles, CA, USA. In industry, he was with LG Electronics Seocho R\&D Campus (Seoul, Korea, 2006--2009), InterDigital (San Diego, CA, USA, 2012), and Intel Corporation (Santa Clara, CA, USA, 2013--2016). 
	
	He is a senior member of the IEEE; and a member of IEEE Communications Society. He was awarded Annenberg Graduate Fellowship with his Ph.D. admission from USC (2009).\end{IEEEbiography}
\begin{IEEEbiography}[{\includegraphics[width=1in,height=1.25in,clip,keepaspectratio]{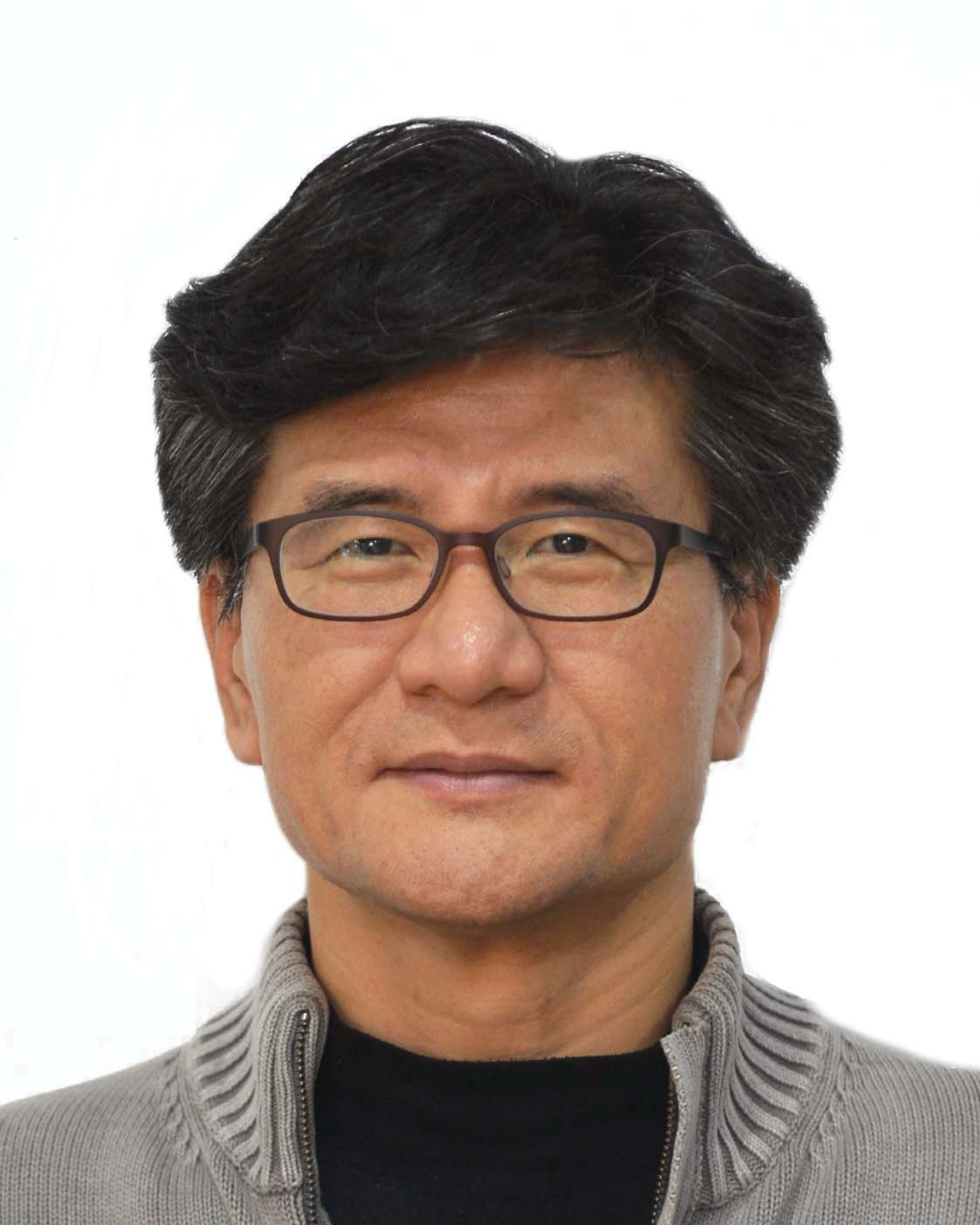}}]{Jaekyun Moon}
	received the Ph.D degree in electrical and computer engineering at Carnegie Mellon University, Pittsburgh, Pa, USA. He is currently a Professor of electrical enegineering at KAIST. From 1990 through early 2009, he was with the faculty of the Department of Electrical and Computer Engineering at the University of Minnesota, Twin Cities. He consulted as Chief Scientist for DSPG, Inc. from 2004 to 2007. He also worked as Chief Technology Officier at Link-A-Media Devices Corporation. His research interests are in the area of channel characterization, signal processing and coding for data storage and digital communication. Prof. Moon received the McKnight Land-Grant Professorship from the University of Minnesota. He received the IBM Faculty Development Awards as well as the IBM Partnership Awards. He was awarded the National Storage Industry Consortium (NSIC) Technical Achievement Award for the invention of the maximum transition run (MTR) code, a widely used error-control/modulation code in commercial storage systems. He served as Program Chair for the 1997 IEEE Magnetic Recording Conference. He is also Past Chair of the Signal Processing for Storage Technical Committee of the IEEE Communications Society. He served as a guest editor for the 2001 IEEE JSAC issue on Signal Processing for High Density Recording. He also served as an Editor for IEEE TRANSACTIONS ON MAGNETICS in the area of signal processing and coding for 2001-2006. He is an IEEE Fellow.
\end{IEEEbiography}\vfill




\end{document}